\documentclass[11pt]{article}
\usepackage{amsfonts
}
\usepackage{amsmath,amsthm,amssymb}

  \newtheorem{proposition}{Proposition}[section]
\usepackage{graphicx}
\usepackage{color}
\usepackage{braket}
\usepackage[utf8]{inputenc}
\usepackage{hyperref} 

\pdfoutput=1 
\usepackage{tikz}
\usetikzlibrary{shapes.geometric,through,decorations.pathmorphing,decorations.markings,arrows.meta,quotes,mindmap,shapes.symbols,shapes.arrows,automata,angles,3d,trees,shadows,automata,arrows,
shapes.callouts}

\def \be {\begin{equation}}
\def \ee {\end{equation}}
\def \ba {\begin{aligned}}
\def \ea {\end{aligned}}
\def \bea {\begin{eqnarray}}
\def \eea {\end{eqnarray}}

\newtheorem{lemma}{Lemma}[section]

\theoremstyle{definition}

\newtheorem{remark}[lemma]{Remark}

\newtheorem{theorem}[lemma]{Theorem}

\newcommand{\hphi}{\widehat{\varphi}}
\newcommand{\hpsi}{\widehat{\psi}}
\newcommand{\SA}{
  {{\Delta}} \hspace{-.6em}
  \raisebox{.37ex}{\resizebox{.6ex}{!}{$\mathrm{/}$}}
  \hspace{.4em} {}
}

\begin{document}
\begin{titlepage}
\begin{flushright}
\end{flushright}
\vspace{0.5cm}
\begin{center}
{\Large \bf Resurgence for large $c$ expansion in Coulomb gas formalism}
\lineskip .75em
\vskip 2.5cm
{Yong Li$^{a,}\footnote{liyong@bimsa.cn}$ and Hongfei Shu$^{b,a,c,}$\footnote{shuphy124@gmail.com, shu@zzu.edu.cn}
}
\vskip 2.5em
 {\normalsize\it 
$^{a}$Beijing Institute of Mathematical Sciences and Applications, Beijing, 101408, China\\
$^{b}$Institute for Astrophysics, School of Physics,
Zhengzhou University, Zhengzhou, Henan 450001, China\\
$^{c}$Yau Mathematical Sciences Center, Tsinghua University, Beijing, 100084, China
}
\vskip 3.0em
\end{center}

\begin{abstract}
We develop a resurgence analysis for large central charge (large $C$) expansions in two-dimensional CFTs using the Coulomb gas formalism. Through the exact Borel-Laplace representations of the conformal blocks $I_1(C,z)$ and $I_2(C,z)$ associated with the four-point correlation function $ \langle \phi_{2,1}(0)\phi_{2,1}(z,\bar{z})\phi_{2,1}(1)\phi_{2,1}(\infty)\rangle$, we demonstrate that $I_1(C,z)$ participates in the Stokes phenomenon of $I_2(C,z)$ (and vice versa), and establish that monodromy in $z$ arises from alien calculus in the Borel plane variable $\zeta$ (Borel dual to $C$). From a given conformal block, resurgence theory thus enables us to discover other
internal operators (conformal blocks). This approach establishes a non-perturbative connection between conformal blocks,  shedding light on the resurgence phenomena in more general quantum field theories.
\end{abstract}

\end{titlepage}

\tableofcontents
\newpage

\section{Introduction}
A central task in quantum field theory is to understand the non-perturbative contributions. While various non-perturbative methods, such as integrability \cite{Beisert:2010jr}, localization \cite{Pestun:2016zxk}, and holography \cite{Maldacena:1997re}, have been developed, their applicability is often restricted to highly specialized scenarios requiring additional symmetries or dualities. In recent years, resurgence theory---invented by Jean \'{E}calle in the 1980s \cite{Ecalle81,Ecalle81b,Ecalle85c}---has attracted significant interest in non-perturbative approaches. This method is based on a key feature of quantum field theory: the perturbative series of observables are typically divergent \cite{Dyson:1952tj,GZ-1990}. Resurgence theory provides a pathway to transform these divergences into physically meaningful quantities. Moreover, it enables us to extract the non-perturbative contributions from the detailed structures of perturbative series. For comprehensive reviews and relevant references, see \cite{Marino:2012zq,Dorigoni:2014hea,Aniceto:2018bis}.

However, obtaining a closed analytic form for the perturbative series remains a challenging task in most physically interesting quantum field theories. Similarly, computing such perturbative series to a high enough order for numerical analysis presents difficulties. On the other hand, conformal field theory (CFT), playing a crucial role in the study of statistical mechanics, string theory, and quantum gravity, offers a more tractable setting \cite{DiFrancesco:1997nk}. In particular, the infinite-dimensional conformal symmetry of two-dimensional CFTs enables a systematic framework for analyzing the operator content and constraining the correlation functions. Furthermore, correlation functions involving degenerate primary operators can be determined via the Belavin-Polyakov-Zamolodchikov (BPZ) equation \cite{Belavin:1984vu}. It is thus an interesting question to investigate the resurgence structures in two-dimensional CFTs, as it may shed light on the resurgence phenomena in more general quantum field theories.

One of the most important perturbative series in CFT is the large central charge $c$ expansion, which corresponds to the semi-classical expansion in the context of AdS/CFT correspondence. In \cite{Harlow:2011ny}, the semi-classical limit of the path integral in the Liouville theory has been studied, where the complex saddle points have been found. Moreover, the Stokes phenomenon and the analytic continuation have been applied to the correlation function to explore the theory beyond the physical region. Moreover, based on the DOZZ formula \cite{Dorn:1994xn,Zamolodchikov:1995aa}, it has been shown that the large $c$ expansion of the three-point function is asymptotic \cite{Benjamin:2024cvv}. The singularities on the Borel plane correspond to the complex saddle points found in \cite{Harlow:2011ny}.

To explore the higher-point function of two-dimensional CFT, one usually uses the conformal block expansion of the correlation function. Recently, it was proposed that the large $c$ expansion of these conformal blocks in the (non-unitary) minimal model can be written as a trans-series \cite{Benjamin:2023uib}. Given the asymptotic series of a certain conformal block, the contribution of other conformal blocks can be extracted through the Borel singularity structure \footnote{Intuitively, the asymptotic series of a conformal block can be analogized to a saddle-point expansion. Through the resurgence theory, one can see the connection among the different saddle points/conformal blocks \cite{Berry:1991but}.}. In particular, focus has been placed on the conformal block expansion for the correlation function involving heavy primary degenerate operators $\phi_{2,1}$:
\begin{equation}
    \langle \phi_{2,1}(0)\phi_{2,1}(z,\bar{z})\phi_{2,1}(1)\phi_{2,1}(\infty)\rangle=\sum_{i=1,2}\frac{C_i^2}{z^{2h_{2,1}}\bar{z}^{2\bar{h}_{2,1}}}F(h_i,z)\bar{F}(\bar{h}_i,\bar{z}),
\end{equation}
where $C_i$ are the structure constants, $h_1=h_{1,1}=0$, $h_2=h_{3,1}=-2C-1$ and $h_{2,1}=-3C/4 -1/2$ are the conformal dimensions of the degenerate operator $\phi_{h_i}$\footnote{The central charge of this CFT is $c=13+6(C+1/C)$. In the region of large $c$ (or large $C$), this CFT with operator $\phi_{r,s}$ is non-unitary.}.
$F(h_i,z)$ is the conformal block with internal operator $\phi_{h_i}$. Given the large $c$ expansion of the conformal block $F_{h_{1,1}}$, one can find the contribution of the other operator $\phi_{3,1}$, which is encoded in the operator product expansion (OPE) 
\begin{equation}
    \phi_{2,1}\times \phi_{2,1}=\phi_{1,1}+\phi_{3,1}.
\end{equation}
This resurgent structure has been studied in more detail based on the hypergeometric function expression of $F_{h_i}$ in \cite{Bissi:2024wur}. On the other hand, correlation functions of degenerate operators are more commonly expressed using the Coulomb gas formalism \cite{Dotsenko:1984nm,Dotsenko:1984ad}. In particular, the conformal block $F(h_i,z)$ can be expressed by using 
\begin{equation}
        I_{1}(C,
        z)=
\int_{1}^{\infty} Q^{C} dw
\end{equation}
for $-1<\mathfrak{Re}C<-\frac13$ and 
\begin{equation}\label{equationI2definition}
        I_2 (C,
        z)
        =\int_{0}^{z} Q^{C} dw
\end{equation}
for $-1<\mathfrak{Re}C$. Here $Q(w):= w(w-1)(w-z)$. $I_1(C,z)$ and $I_2(C,z)$ are the conformal blocks associated to the internal operators $\phi_{1,1}$ and $\phi_{3,1}$, respectively. For more general regions of $C$, we could use the Pochhammer contour on the $w$-plane instead \cite{DiFrancesco:1997nk}. By imposing the $U(1)$ charge neutrality, one can also express the higher-point correlation function involving more general degenerate operators in this formalism. 

In this paper, we investigate the resurgent structure of the conformal blocks, $I_1(C,z)$ and $I_2(C,z)$, within the framework of the Coulomb gas formalism. This approach is supposed to be extended to higher-point correlation functions that involve more general degenerate operators. 

Our main contribution is the complete resurgence analysis of the functions \(I_1(C,z)\) and \(I_2(C,z)\) as $C$ goes to infinity. 
By initially fixing \(z\) in a neighborhood of \(\frac{1}{2}\), we establish explicit Borel-Laplace representations for \(I_1(C,z)\) and \(I_2(C,z)\) in Propositions~\ref{propositionI1BL} and~\ref{propositionI2BL}:
\begin{equation}
    I_1(C,z) = Q_+^C \frac{e^{-2\pi i C}}{1-e^{-2\pi i C}} \mathfrak{L}^{\frac\pi2+\varepsilon} \Phi(Q_+e^{-\zeta},z), \quad I_2(C,z) = Q_-^C \mathfrak{L}^{0
    } \Psi(Q_-e^{-\zeta},z),
\end{equation}
where the Borel germs $\Phi$ and $\Psi$ are determined by the solutions $w_i,i=0,1,2$ to $Q=0$ (see equation \eqref{equationsolutionswz}) and $Q_\pm$ are the critical values of the function $Q(w)$ (see equation \eqref{equationcriticalvalueQpm}). The asymptotic behavior for large \(C\) is thus directly recoverable. 
In Section~\ref{sectionalien}, we analyze the alien operators —equivalently characterized by monodromy computations since all singularities are integrable—and explicitly compute the Stokes phenomenon for \(I_2(C,z)\) (Theorem~\ref{theoremI2Stokes}):
\begin{equation}\label{equationintroductionstokes2}
        I_2(C,z)=Q_-^C \frac{1+e^{-2\pi iC}}{1-e^{-2\pi iC}}\mathfrak{L}^{\frac\pi2+\varepsilon} \Psi(Q_-e^{-\zeta},z) +I_1(C,z)  , \quad \mathfrak{Im} C<0
\end{equation}\label{equationintroductionstokes1}
and for \(I_1(C,z)\) (equation \eqref{equationStokesI1}):
\begin{equation}
    I_1(C,z) = Q_+^C\frac{\sin(2\pi C)}{\sin(3\pi C)} e^{-\pi iC} \mathfrak{L}^0 \Phi(Q_+e^{-\zeta},z) - \frac{\sin(\pi C)}{\sin(3\pi C)}  I_2(C,z).
\end{equation}
The above formulas can be analytically extended to $C<0$ or $C>0$, respectively. Our analysis reveals that the Stokes phenomenon in one conformal block intrinsically governs that of another, a connection heuristically posited in works such as \cite{Benjamin:2023uib} and \cite{Bissi:2024wur}.

We then establish in Section~\ref{subsection1-z} a duality between Borel germs $\Psi$ and $\Phi$ appearing in equations~\eqref{equationintroductionstokes2} and~\eqref{equationintroductionstokes1} (Proposition~\ref{proposition1-z}):
\begin{equation}
    \Psi(Q_-e^{-\zeta},1-z) = - \Phi(Q_+e^{-\zeta},z),
\end{equation}
which is supposed to relate the crossing symmetry of the four-point correlation function.

Monodromy computation provides an alternative method to derive $I_1(C,z)$ from $I_2(C,z)$ and vice versa (as noted in \cite{GH24}). At the end of this paper, we compute the monodromy of \(I_1(C,z)\) and \(I_2(C,z)\) in the \(z\)-variable (at \(z = 0\) and \(z = 1\)) by tracking singular-point trajectories on the Borel plane (Figure~\ref{figurehpsihphimonodromy}). This reveals a fundamental relation between \(z\)-monodromy and alien calculus in the \(\zeta\)-variable —the Borel dual of \(C\)— as detailed in Lemma~\ref{lemmaDeltaalienhphi}. This relation is intrinsically linked to co-equational resurgence (termed parametric resurgence in \cite{ecalle1994weighted}) and the WKB analysis of Schr\"odinger-type equations \cite{ATT21}.  
This work serves as an accessible entry point to this technical framework, elucidating connections through explicit computations.

This paper is organized as follows: In
section \ref{sectionBorelLaplaceI1I2}, we provide the Borel-Laplace expressions of $I_1(C,z)$ and $I_2(C,z)$ for $z$ around $\frac12$. In section \ref{sectionalien}, the Alien calculus on Borel germs of $I_1(C,z)$ and $I_2(C,z)$ and the Stokes phenomenon of them are presented. Based on the Alien calculus, we evaluate the monodromy computation for variable $z$ in section \ref{sectionmonodromy}. In section \ref{sec:conclusion}, we will conclude and present possible
future directions.

\section{Borel-Laplace expressions on conformal blocks $I_1,I_2$}\label{sectionBorelLaplaceI1I2}

In order to write $I_1,I_2$ to be Laplace integrals, we need to establish the following formula
\begin{equation}
    \int Q^Cdw = \int Q^C f(Q) dQ,
\end{equation}
where \( f(Q) \) consists of the solutions \( w_i(Q)\) (with $\ i=0,1,2$) to the equation \( w(w-1)(w-z) - Q = 0 \). 
We mainly focus on the case as $z$ nears $1/2$ in this section. Other regions of $z$ can be obtained via the analytical continuation. Let \[\boxed{z \sim \frac12 .}\]

\subsection{The solutions to $Q=w(w-1)(w-z)$}

The critical points of $Q(w)=w(w-1)(w-z)$ are located at 
\begin{equation}
    w_{\pm}(z)=\frac13(z+1) \pm \frac13 (z^2-z+1)^{\frac12} \ \sim \frac12 \pm \frac{\sqrt3}{6}. 
\end{equation}
The corresponding critical values are
\begin{equation}\label{equationcriticalvalueQpm}
    Q_{\pm}(z):=Q(w_{\pm}(z)) \ \sim \mp\frac{\sqrt3}{36}. 
\end{equation}
Moreover, the solutions to $w(w-1)(w-z)-Q=0$ are (see formula (5) of  \cite{Nickalls1993ANA})
\begin{equation}\label{equationsolutionswz}
    \begin{aligned}
        w_{0}(Q,z)
        =&
       \frac{z+1}3+2\left(-\frac{u}{3}\right)^{\frac12}\cos\left(\frac{1}{3}\arccos\left(\frac{3v}{2u}\left(-\frac{3}{u}\right)^{\frac12}\right)+\frac{2\pi}{3}\right),
        \\
       w_{1}(Q,z)
        =&
       \frac{z+1}3+2\left(-\frac{u}{3}\right)^{\frac12}\cos\left(\frac{1}{3}\arccos\left(\frac{3v}{2u}\left(-\frac{3}{u}\right)^{\frac12}\right)\right),
        \\
        w_{2}(Q,z)
        =&
       \frac{z+1}3+2\left(-\frac{u}{3}\right)^{\frac12}\cos\left(\frac{1}{3}\arccos\left(\frac{3v}{2u}\left(-\frac{3}{u}\right)^{\frac12}\right)-\frac{2\pi}{3}\right),
        \\
    \end{aligned}
\end{equation}
where $u$ and $v$ are defined by
\begin{equation}\label{equationuv}
u=\frac{3z-(z+1)^{2}}{3},\quad v=\frac{1}{27}\big(-2(z+1)^{3}+9z(z+1)-27Q\big).
\end{equation}

We regard \( w_i(Q,z) \), for \( i = 0,1,2 \), as holomorphic germs defined in a neighborhood of \(\{0\} \times \{\frac{1}{2}\}\) in the variables \((Q,z)\). The initial conditions in \(Q\) are given by
\begin{equation}
    w_0(0,z) = 0, \quad
    w_1(0,z) = 1, \quad
    w_2(0,z) = z,
\end{equation}
which correspond precisely to the solutions of \( w(w-1)(w-z) = 0 \) at \( Q = 0 \). As $z$ near $\frac12$, the \( w_i \) (for \( i = 0,1,2 \)) are resurgent functions possessing singularities at \( Q_{\pm} \). More precisely, the singular part of each \( w_i \) arises from the function \( \arccos \) at \( \pm 1 \), producing the possibly\footnote{``possibly'' since some points are regular, see equation \eqref{equationw0w1holomorphic}} singular set
\begin{equation}\label{equationsingularsetsinducesbyQ}
    \left\{(Q,z) \mid \frac{3v}{2u}\left(-\frac{3}{u}\right)^{\frac{1}{2}} = \pm 1 \right\} = \{(Q,z) \mid 4u^3 + 27 v^2 = 0\} = \{(Q,z) \mid Q = Q_{\mp}(z)\}.
\end{equation}

For fixed $z$ around $\frac12$, by using the analytic continuation of function $\arccos$ computed in Appendix \ref{Appendixarccos} (Lemma \ref{lemmaarccosanalyticcontinuation}), we have the analytic continuation of these solutions $\omega_i$'s in variable $Q$
\begin{equation}\label{equationanalyticityofwz}
\begin{split}
{\rm cont}_{Q_+(z)} w_1 (Q,z) & = w_2 (Q,z), \quad
{\rm cont}_{Q_+(z)} w_2 (Q,z)= w_1 (Q,z), 
\\
{\rm cont}_{Q_-(z)} w_{0} (Q,z) &= w_{2} (Q,z), \quad
{\rm cont}_{Q_-(z)} w_{2} (Q,z)= w_{0} (Q,z)     
\end{split}
\end{equation}
 and 
\begin{equation}\label{equationw0w1holomorphic}
    w_0(Q,z) \text{ holomorphic at } (Q_+(z),z),\quad
    w_1(Q,z) \text{ holomorphic at } (Q_-(z),z).
\end{equation}
Moreover, 
\begin{equation}
w_2(Q_{\pm},z) = w_{\pm}, \quad
w_1(Q_+,z)=w_+ \neq w_0(Q_+,z), \quad
w_0(Q_-,z)=w_- \neq w_1(Q_-,z).
\end{equation}
See Figure \ref{figuretriplecover} for the evaluations in different leaves.


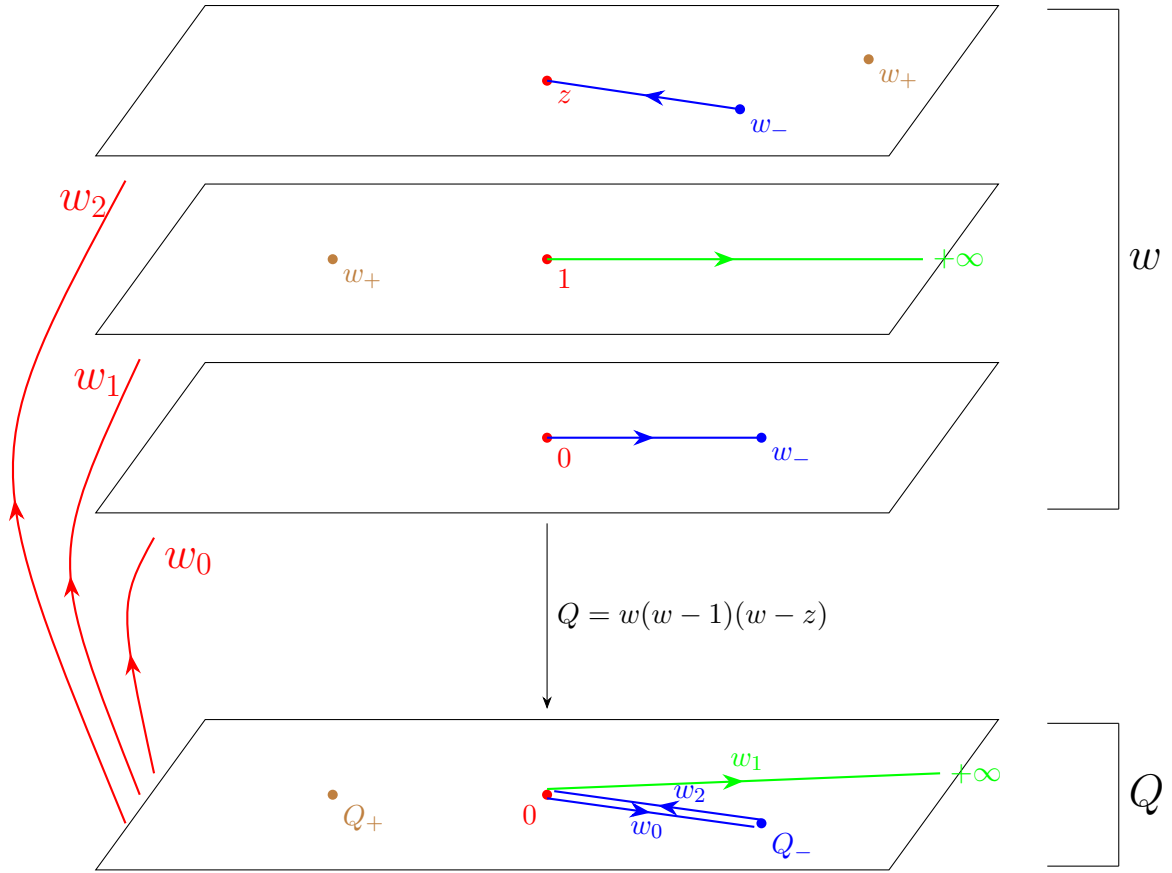
\begin{figure}[h] 
    \centering
\begin{tikzpicture}[scale=0.95,
    node distance = 2.5cm,
    paraplane/.style = { 
        draw,
        shape=trapezium,
        trapezium left angle=10,
        trapezium right angle=170,
        trapezium stretches body=true,
        minimum width=12cm,
        minimum height=2cm,
        text centered
    },
    arrow/.style = {red, thick, -Latex, bend right=45},
    dot/.style = {circle, fill, inner sep=0.5pt},
        mid arrow/.style={
        postaction={decorate}, 
        decoration={
            markings,
            mark=at position 0.5 with {\arrow{Stealth[scale=1.2]}} 
        }
    }
]

\node[paraplane] () at (0,6) {};
\node[paraplane] () at (0,3.5) {};
\node[paraplane] () at (0,1) {};
\node[paraplane] () at (0,-4) {};

\fill[red] (0,-4) circle (2pt) node[below left] {$0$};
\fill[red] (0,1) circle (2pt) node[below right] {$0$};
\fill[red] (0,3.5) circle (2pt) node[below right] {$1$};
\fill[red] (0,6) circle (2pt) node[below right] {$z$};

\draw[-Stealth](0,-0.2)--(0,-2.8); \node[right] at (0,-1.5) {$Q = w(w-1)(w-z)$};

\draw(7,7)--(8,7);
\draw(8,7)--(8,0);
\draw(8,0)--(7,0);
\node[right] at (8,3.5) {\LARGE $w$};

\draw(7,-3)--(8,-3);
\draw(8,-3)--(8,-5);
\draw(8,-5)--(7,-5);
\node[right] at (8,-4) {\LARGE $Q$};

\fill[blue] (3,1) circle (2pt) node[below right] {$w_-$};
\fill[blue] (2.7,5.6) circle (2pt) node[below right] {$w_-$};
\fill[blue] (3,-4.4) circle (2pt) node[below right] {$Q_-$};
\fill[brown] (-3,3.5) circle (2pt) node[below right] {$w_+$};
\fill[brown] (4.5,6.3) circle (2pt) node[below right] {$w_+$};
\fill[brown] (-3,-4) circle (2pt) node[below right] {$Q_+$};

\draw[thick, red, mid arrow] 
    (-5.5,-3.7) .. controls (-6,-1.4) and (-6,-1.3) .. (-5.5,-0.4);
\draw[thick, red, mid arrow] 
    (-5.7,-4) .. controls (-7,-0.7) and (-7,-0.7) .. (-5.7,2.1);
\draw[thick, red, mid arrow] 
    (-5.9,-4.4) .. controls (-8,0.7) and (-8,0.7) .. (-5.9,4.6);
\node[below right,red] at (-5.5,-0.4) {\LARGE $w_0$};
\node[below left,red] at (-5.8,2.1) {\LARGE $w_1$};
\node[below left,red] at (-6,4.6) {\LARGE $w_2$};

\draw[thick, blue, mid arrow](0,-4.05)--(2.9,-4.45);
\node[below,blue] at (1.4,-4.23) {$w_0$};
\draw[thick, blue, mid arrow](3,-4.35)--(0.1,-3.95);
\node[above,blue] at (2,-4.25) {$w_2$};

\draw[thick, blue, mid arrow](0,1)--(3,1);
\draw[thick, blue, mid arrow](2.7,5.6)--(0,6);
\draw[thick, green, mid arrow] (0,3.5)--(5.26,3.5);
\draw[thick, green, mid arrow] (0,-3.92)--(5.5,-3.7);
\node[right,green] at (5.26,3.5) {$+\infty$};
\node[right,green] at (5.5,-3.7) {$+\infty$};
\node[above,green] at (2.8,-3.8) {$w_1$};
\end{tikzpicture}
\caption{
When $z$ lies in a small neighborhood of $\frac{1}{2}$, the cubic projection $Q(w) = w(w-1)(w-z)$ induces a local triple cover. The fibers are parameterized by three maps $w_i\ (i=0,1,2)$. Proposition \ref{propositionI2BL} expresses the integral $I_2(C,z)$ as an integral over the $Q$-plane along the blue line segments, which are preimages of $w_0$ and $w_2$. Similarly, Proposition \ref{propositionI1BL} represents $I_1(C,z)$ via integration over the $Q$-plane along the green ray, which is the preimage of $w_1$.
} 
    \label{figuretriplecover} 
\end{figure}

\subsection{The Borel-Laplace transform of $I_1$ and $I_2$}             

In this section, we express the conformal blocks $I_1,I_2$ as the forms of the Borel-Laplace transform.
By the discussion in the previous section, a direct computation yields

\begin{equation}\label{equationI2Laplacecomputation}
    \begin{aligned}
        I_{2}(C,z)
        =&\int_{0}^{z}dw\Big[w(w-1)(w-z)\Big]^{C}
        \\
        =&
        \int_{0}^{Q_-}dQ Q^{C}\Big[\frac{1}{\frac{dQ}{dw}\Big|_{w_{0}(Q)}}-\frac{1}{\frac{dQ}{dw}\Big|_{w_{2}(Q)}}\Big]
        \\
        =&
        \int_{0}^{Q_-}dQ \ Q^{C} \Big[\frac{1}{(w_0-w_1)(w_0-w_2)}-\frac{1}{(w_2-w_0)(w_2-w_1)}\Big]
        \\
        =&
        \int^{+\infty}_{\zeta_-}e^{-\zeta C} \left.\frac{2w_1-w_0-w_2}{(w_0-w_1)(w_0-w_2)(w_1-w_2)}\right|_{Q=e^{-\zeta}} e^{-\zeta} d\zeta,
    \end{aligned}
\end{equation}
where $\zeta_- = -\log Q_-$ in the last formula. See the blue lines in Figure \ref{figuretriplecover}. Let us write $I_2$ as a form of Laplace integral and combine this result and the previous calculations into a proposition.

\begin{proposition}\label{propositionI2BL}
    Let $Q(w)=w(w-1)(w-z)$ and $I_2(C,z):=\int_0^z Q^Cdw$. Let $z$ near $\frac12$ and let $Q_{\pm}(z)$ (see \eqref{equationcriticalvalueQpm}) be the critical values of the polynomial $Q(w)$. Then we have
\begin{equation}\label{equationI2laplace}
    I_2(C,z) = Q_-^{C} \mathfrak{L}^0 \Psi(Q_-e^{-\zeta}),
\end{equation}
where
$\mathfrak{L}^\theta$ is the Laplace transform along direction $\theta$: $\mathfrak{L}^\theta \bullet:=\int_0^{e^{i\theta}\infty} e^{-C\zeta}\bullet d\zeta$ and the germ
\begin{equation}
    \Psi(Q) := Q\cdot \left(\frac{2w_1-w_0-w_2}{(w_0-w_1)(w_0-w_2)(w_1-w_2)}\right)(Q)
\end{equation}
is defined in terms of $w_i$ by equations in  \eqref{equationsolutionswz}. Moreover, the asymptotic behavior of $I_2$ as $C$ goes to infinity with $ -\pi + \varepsilon <\arg C < \pi -\varepsilon$ is 
    \begin{equation}\label{equationI2asymp}
        I_2(C,z) \sim Q_-^{C} \mathfrak{B}^{-1} \big(\Psi(Q_-e^{-\zeta})\big)
    \end{equation}
    where $\mathfrak{B}$ is the Borel transform, which is an isomorphism between power series: 
    \begin{equation}
        \mathfrak{B}:\sum\limits_{n\geq0}a_n(z)C^{-n-\beta} \mapsto \sum\limits_{n\geq0}\frac{a_n(z)}{\Gamma(n+\beta)}\zeta^{n+\beta-1}\quad {\rm with}\quad 0<\beta<1.
    \end{equation}
\end{proposition}

\begin{proof}
    Formula \eqref{equationI2laplace} is proved in the computation \eqref{equationI2Laplacecomputation}. Notice that for any fixed $z\neq 0,1$, the Borel germ $\Psi(Q_-e^{-\zeta})$ growth like $e^{-\zeta}$ as $\zeta$ goes to $+\infty$. By the classical result in Borel-Laplace method (see for instance, the chapter $5$ of \cite{CS16}), $I_2$ in the viewpoint of \eqref{equationI2laplace} is holomorphic as $\mathfrak{Re} C>-1$ for all $z\neq 0,1$. The asymptotic behavior \eqref{equationI2asymp} as $C$ goes to infinity with $ -\pi + \varepsilon <\arg C < \pi -\varepsilon$ is a directly result of \eqref{equationI2laplace}.
\end{proof}

\begin{remark}
Rewriting the integrand by $Q^C=e^{-C \big(-\log Q\big)}$, the asymptotic expansion of $I_2$ 
can be obtained using the saddle point expansion around $z_\ast=\frac{1}{3} (1+z-\sqrt{z^2-z+1})$, where $z_\ast$ is saddle point of $-\log Q$ locating at the interval $[0,z]$. At $z=1/2$, the saddle point expansion of $I_2(C,z)$ is 
\begin{equation}
    Q_{-}^{-C}I_{2}(C,z)\sim\sqrt{\frac{\pi}{2}}\Big(\frac{1}{3}C^{-\frac{1}{2}}-\frac{11}{108}C^{-\frac{3}{2}}+\frac{265}{7776}C^{-\frac{5}{2}}-\frac{5075}{839808}C^{-\frac{7}{2}}+\frac{261457}{120932352}C^{-\frac{9}{2}}+\cdots\Big),
\end{equation}
which coincides with the Borel inverse transform $\mathfrak{B}^{-1} \big(\Psi(Q_-e^{-\zeta})\big)$ in \eqref{equationI2asymp} at $z=1/2$. Similar checks can be done in the region $0<z<1$.
\end{remark}
For $I_1$, as $z$ around $\frac12$, a directly computation yields
\begin{equation}
    \begin{aligned}
        I_{1}(C,z)
        =&\int_{1}^{+\infty}dw\Big[w(w-1)(w-z)\Big]^{C}
        \\
        =&
        \int_{0}^{+\infty}dQ Q^{C}\Big[\frac{1}{\frac{dQ}{dw}\Big|_{w_{1}(Q)}}\Big]
        \\
        =&
        \int_{0}^{+\infty}dQ \ Q^{C} \Big[\frac{1}{(w_1-w_0)(w_1-w_2)}\Big].
    \end{aligned}
\end{equation}

We used 
\begin{equation}
    \lim\limits_{Q\rightarrow+\infty} w_1(Q) = +\infty \quad \text{as } z \text{ near } \frac12
\end{equation}
in the second equality. This can be proved by using the equation \eqref{equationarccos}.
\begin{lemma}\label{lemmaHpropoty}
    For the germ $H(Q):=\left(\frac{1}{(w_1-w_0)(w_1-w_2)}\right)(Q)$, the singular point of this germ in variable $Q$ is $Q_+(z) $. Moreover, $H(Q)$ is integrable and has a double-branch at $Q_+$.
\end{lemma} 

\begin{proof}
    Let $z$ near $\frac12$. A direct computation yields 
    \begin{equation}\label{equationinlemmaH}
        H(Q)=\frac29 \big(-\frac{u}3\big)^{-\frac12}\frac{1}{1-\frac43\sin^2y},
    \end{equation}
    where $y:=\frac13\arccos\left(\frac{3v}{2u}\left(-\frac{3}{u}\right)^{\frac12}\right)$ and $u,v$ are defined in equation \eqref{equationuv}. Notice that $\frac{3v}{2u}\left(-\frac{3}{u}\right)^{\frac12}$ is linear in $Q(z)$. Together with the discussion in equation \eqref{equationsingularsetsinducesbyQ}, we have that the possibly singular points of $H$ in $Q$ are $Q_{\pm}(z)$. Since $({\rm cont}_{Q_{\pm}})^2w_i=w_i$ for $i=0,1,2$, it is clear that 
    \begin{equation}
        ({\rm cont}_{Q_\pm})^2 H = H.
    \end{equation}The monodromy computation at $Q_\pm$ are, by using equations \eqref{equationanalyticityofwz} and \eqref{equationw0w1holomorphic},
\begin{equation}\label{equationmonodromyH}
\begin{split}
    \Delta_{Q_+}H:=&\left(\frac{1}{(w_1-w_0)(w_1-w_2)}\right)(Q)- \left({\rm cont}_{Q_+} \frac{1}{(w_1-w_0)(w_1-w_2)}\right)(Q) 
    \\
    =&
    \left(\frac{1}{(w_1-w_0)(w_1-w_2)}\right)(Q)
    -\left(\frac{1}{(w_0-w_2)(w_1-w_2)}\right)(Q)
    \\
    =&
    \left(\frac{w_1+w_2-2w_0}{(w_0-w_1)(w_0-w_2)(w_1-w_2)}\right)(Q).
    \\
    \Delta_{Q_-}H:=&\left(\frac{1}{(w_1-w_0)(w_1-w_2)}\right)(Q)- \left({\rm cont}_{Q_-} \frac{1}{(w_1-w_0)(w_1-w_2)}\right)(Q)=0.
\end{split}
\end{equation}
Moreover, by formula \eqref{equationinlemmaH}, one can show that the limit as $Q\rightarrow Q_\pm$ in principle sheet are
\begin{equation}
    \lim\limits_{Q\rightarrow Q_+} (Q-Q_+)^{\frac12} H(Q) < \infty \quad \text{and} \quad \lim\limits_{Q\rightarrow Q_-} H(Q) < \infty .
\end{equation}
Indeed, the existence of the limit is implied by the direct computations of the existence of the limit of 
\begin{equation}
    \lim\limits_{x\rightarrow -1}\frac{(x+1)^{\frac12}}{1-\frac43 \sin^2(\frac13\arccos x)} \quad \text{and}  \quad
    \lim\limits_{x\rightarrow 1}\frac{1}{1-\frac43 \sin^2(\frac13\arccos x)}
\end{equation} 
and the linearity between $x=\frac{3v}{2u}\left(-\frac{3}{u}\right)^{\frac12}$ and $Q$. We conclude that $H$ is regular at $Q_-$ and has an integrable double-branch singular point at $Q_+$.
\end{proof}
We then write $I_1$ to be a form of the Laplace transform on the real axis.
\begin{equation}
    I_1(C,z)=\int^{+\infty}_{-\infty}e^{-\zeta C} \left.\frac{1}{(w_1-w_0)(w_1-w_2)}\right|_{Q=e^{-\zeta}} e^{-\zeta} d\zeta.
\end{equation}
The singular points $\xi_m$ in variable $\zeta$ are determined by equation $Q_+=e^{-\zeta}$
\begin{equation}
    \xi_m := -\log Q_+ +2\pi im  \quad\text{with } m\in\mathbb{Z}.
\end{equation}

Notice that $-\log Q_+$ is almost $-\pi i - \log \frac{\sqrt3}{36}$ as $z$ tends to $\frac12$. For $\mathfrak{Im} C<0$,
\begin{equation}
\begin{split}
    I_1(C,z) 
    &=
    \int^{+\infty}_{-\infty}e^{-\zeta C} \left.\frac{1}{(w_1-w_0)(w_1-w_2)}\right|_{Q=e^{-\zeta}} e^{-\zeta} d\zeta
    \\
    &=
    \sum\limits_{m\geq 1} \int_{\gamma_{m}}e^{-\zeta C} \left.\frac{1}{(w_1-w_0)(w_1-w_2)}\right|_{Q=e^{-\zeta}} e^{-\zeta} d\zeta
    \\
    &=
    \sum\limits_{m\geq1} \int_{\xi_m}^{e^{i(\frac\pi2+\varepsilon)} \infty} e^{-C\zeta} \left.\frac{w_1
    +w_2-2w_0}{(w_0-w_1)(w_0-w_2)(w_1-w_2)}\right|_{Q=e^{-\zeta}} e^{-\zeta} d\zeta
    \\
    &=
    Q_+^{C}\sum\limits_{m\geq1} e^{-2\pi imC}\int_{0}^{e^{i(\frac\pi2+\varepsilon)} \infty} e^{-C\zeta} \Phi(Q_+e^{-\zeta})
\end{split}
\end{equation}
where the integral curves $\gamma_m$ in the second line of the above formula are shown in Figure \ref{StokesI1}. The computation from the second line to the third line is due to the monodromy computation in equation \eqref{equationmonodromyH} and the integrability of the singular points $Q_+$, which are all shown in Lemma \ref{lemmaHpropoty}. 

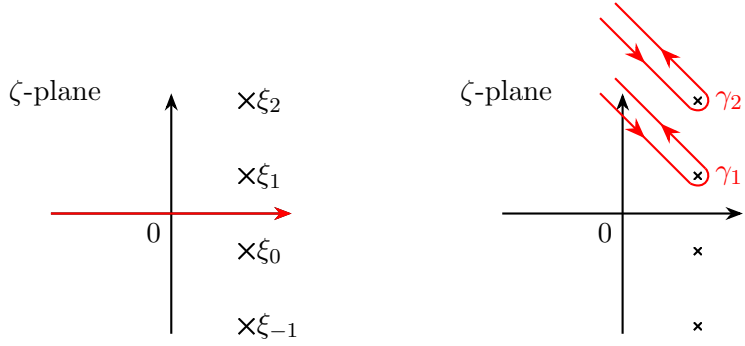
\begin{figure}[h]
\centering
\begin{tikzpicture}[scale=2,        
mid arrow/.style={
        postaction={decorate}, 
        decoration={
            markings,
            mark=at position 0.5 with {\arrow{Stealth[scale=1.2]}} 
        }
    }]

\begin{scope}[xshift=-1cm, local bounding box=left]
  \draw[-{Stealth[scale=1]}, thick] (-0.8,0)--(0.8,0); 
  \draw[-{Stealth[scale=1]}, thick] (0,-0.8)--(0,0.8); 
  \node[inner sep=3pt, path picture={
    \draw[thick] (path picture bounding box.south west) -- (path picture bounding box.north east);
    \draw[thick] (path picture bounding box.north west) -- (path picture bounding box.south east);
  }] at (0.5,-0.25) {};
    \node[inner sep=3pt, path picture={
    \draw[thick] (path picture bounding box.south west) -- (path picture bounding box.north east);
    \draw[thick] (path picture bounding box.north west) -- (path picture bounding box.south east);
  }] at (0.5,-0.75) {};
    \node[inner sep=3pt, path picture={
    \draw[thick] (path picture bounding box.south west) -- (path picture bounding box.north east);
    \draw[thick] (path picture bounding box.north west) -- (path picture bounding box.south east);
  }] at (0.5,0.75) {};
  \node[inner sep=3pt, path picture={
    \draw[thick] (path picture bounding box.south west) -- (path picture bounding box.north east);
    \draw[thick] (path picture bounding box.north west) -- (path picture bounding box.south east);
  }] at (0.5,0.25) {};
  \node[below left] at (0,0) {$0$};
  \node[left] at (-0.4,0.8) {$\zeta$-plane};
  \node[right] at (0.5,0.25) {$\xi_1$};
  \node[right] at (0.5,-0.25) {$\xi_0$};
  \node[right] at (0.5,-0.75) {$\xi_{-1}$};
  \node[right] at (0.5,0.75) {$\xi_2$};
  \draw[-{Stealth[scale=1]}, thick,red] (-0.8,0)--(0.8,0);
\end{scope}

\begin{scope}[xshift=2cm, local bounding box=right, dot/.style = {circle, fill, inner sep=0.5pt}]
  \draw[-{Stealth[scale=1]}, thick] (-0.8,0)--(0.8,0); 
  \draw[-{Stealth[scale=1]}, thick] (0,-0.8)--(0,0.8); 
  \node[inner sep=1.5pt, path picture={
    \draw[thick] (path picture bounding box.south west) -- (path picture bounding box.north east);
    \draw[thick] (path picture bounding box.north west) -- (path picture bounding box.south east);
  }] at (0.5,-0.25) {};
    \node[inner sep=1.5pt, path picture={
    \draw[thick] (path picture bounding box.south west) -- (path picture bounding box.north east);
    \draw[thick] (path picture bounding box.north west) -- (path picture bounding box.south east);
  }] at (0.5,-0.75) {};
    \node[inner sep=1.5pt, path picture={
    \draw[thick] (path picture bounding box.south west) -- (path picture bounding box.north east);
    \draw[thick] (path picture bounding box.north west) -- (path picture bounding box.south east);
  }] at (0.5,0.75) {};
  \node[inner sep=1.5pt, path picture={
    \draw[thick] (path picture bounding box.south west) -- (path picture bounding box.north east);
    \draw[thick] (path picture bounding box.north west) -- (path picture bounding box.south east);
  }] at (0.5,0.25) {};
  \node[below left] at (0,0) {$0$};
  \node[left] at (-0.4,0.8) {$\zeta$-plane};

\node[right,red] at (0.55,0.25) {$\gamma_1$};
  \draw[thick, red, mid arrow](-0.15,0.8)--(0.45,0.2);
  \draw[thick, red, mid arrow](0.55,0.3)--(-0.05,0.9);
\draw[thick,red] (0.45,0.2) .. controls (0.55,0.15) and (0.6,0.25) .. (0.55,0.3);

\node[right,red] at (0.55,0.75) {$\gamma_2$};
  \draw[thick, red, mid arrow](-0.15,1.3)--(0.45,0.7);
  \draw[thick, red, mid arrow](0.55,0.8)--(-0.05,1.4);
  \draw[thick,red] (0.45,0.7) .. controls (0.55,0.65) and (0.6,0.75) .. (0.55,0.8);
 
\end{scope}
\end{tikzpicture}
\caption{Writing the integral along the real axis as an integral on the sum of Hankel contours $\gamma_m$ (for $m\geq1$).}
\label{StokesI1}
\end{figure}

We conclude in the following Proposition.
\begin{proposition}\label{propositionI1BL}
    Let $Q(w)=w(w-1)(w-z)$ and $I_1(C,z):=\int_0^{+\infty} Q^Cdw$. Let $z$ close to $\frac12$ and let $Q_{\pm}(z)$ (see \eqref{equationcriticalvalueQpm}) be the critical values of the polynomial $Q(w)$. Then we have, for sufficient small $\varepsilon$ and for $\mathfrak{Im}C<0$,
\begin{equation}\label{equationI2zlaplace}
    I_1(C,z) = Q_+^C \frac{e^{-2\pi i C}}{1-e^{-2\pi i C}} \mathfrak{L}^{\frac\pi2+\varepsilon} \Phi(Q_+e^{-\zeta}),
\end{equation}
where
\begin{equation}
    \Phi(Q)= Q \cdot\left(\frac{w_1+w_2-2w_0}{(w_0-w_1)(w_0-w_2)(w_1-w_2)}\right)(Q)
\end{equation}
is defined in terms of $w_i$ by equations \eqref{equationsolutionswz}. Thus $I_1$ is exponentially small as $C$ goes to $\infty$ with $\mathfrak{Im} C<0$. More precisely, the trans-series expansion\footnote{There is no precise definition of such terminology. We use it here for a good correspondence between the location of the singular points and the power appeared in the expansion.} of $I_1$ is 
\begin{equation}\label{eq:I1-asy-exp}
    I_1(C,z) \sim Q_+^C \sum\limits_{m\geq1} e^{-2\pi imC} \mathfrak{B}^{-1} \Phi(Q_+e^{-\zeta})
\end{equation}
as $C$ goes to $\infty$ with $\mathfrak{Im} C<0$.
\end{proposition}

\begin{proof}
Similar to the proof of Proposition \ref{propositionI2BL}. 
\end{proof}

\begin{remark}\label{rmk:I1-asym}
The Borel inverse transform in \eqref{equationI2zlaplace} is 
\begin{equation}
    \begin{aligned}
       \mathfrak{L}^{\frac{\pi}{2}+\varepsilon}\Phi(Q_{+}e^{-\zeta})\sim&\sqrt{\frac{\pi}{2}}\Big(-\frac{1}{3}C^{-\frac{1}{2}}+\frac{11}{108}C^{-\frac{3}{2}}-\frac{265}{7776}C^{-\frac{5}{2}}+\frac{5075}{839808}C^{-\frac{7}{2}}-\frac{261457}{120932352}C^{-\frac{9}{2}}+\cdots\Big)\\=:&\sum_{n=0}C^{-n-\frac{1}{2}}s_{n}.
    \end{aligned}
\end{equation}
On the other hand, we can introduce an expansion
\begin{equation}
    Q_{+}^{-C}\frac{1-e^{-2\pi iC}}{e^{-2\pi iC}}I_{1}(C,z)\sim\sum_{n=0}C^{-n-\frac{1}{2}}s_{n}^{\prime},
\end{equation}
by using the following hypergeometric function expression of $I_1$
\begin{equation}\label{eq:I1-hyper}
    I_{1}(z)=\frac{\Gamma(-3C-1)\Gamma(1+C)}{\Gamma(-2C)}{}_{2}F_{1}(-C,-1-3C,-2C,z).
\end{equation}
At large $C$, we can approach $s_n^\prime$ by plotting
\begin{equation}
    s_{n}^{\prime}\sim C^{n+\frac{1}{2}}\Big(Q_{+}^{-C}\frac{1-e^{-2\pi iC}}{e^{-2\pi iC}}I_{1}(C,z)-\sum_{m=0}^{n-1}C^{-m-\frac{1}{2}}s_{m}^{\prime}\Big).
\end{equation}
In Fig.\ref{fig:sn-plot}, we compare $s_n$ and $s_n^\prime$ at $z=1/2$ for pure imaginary negative  $C$. As we can see, $s_n$ and $s_n^\prime$ match well at large $C$. A similar comparison can be done for $0<z<1$.
\begin{figure}
    \centering
    \includegraphics[width=0.45\linewidth]{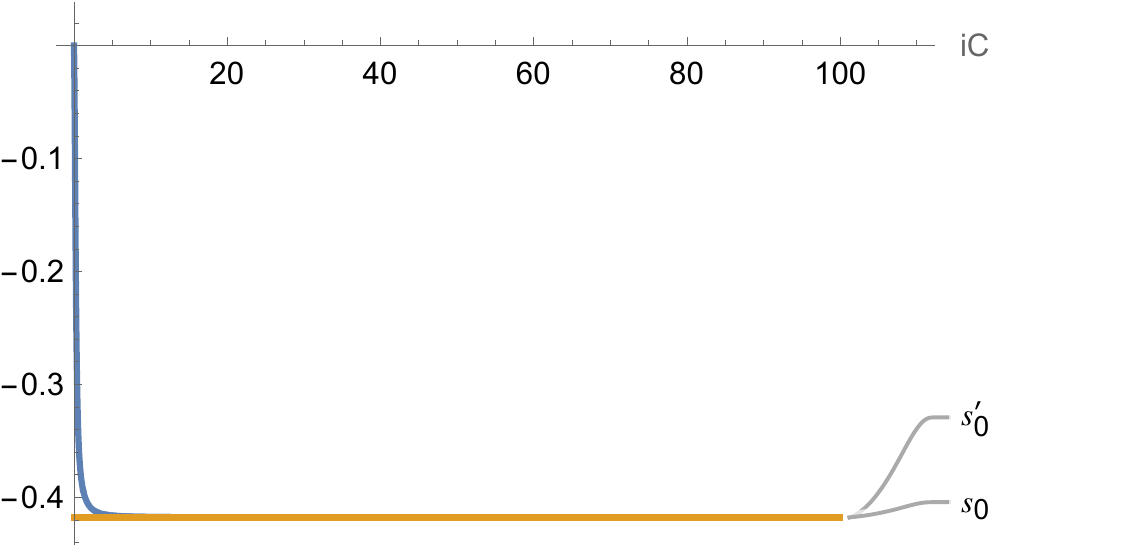}
    \includegraphics[width=0.45\linewidth]{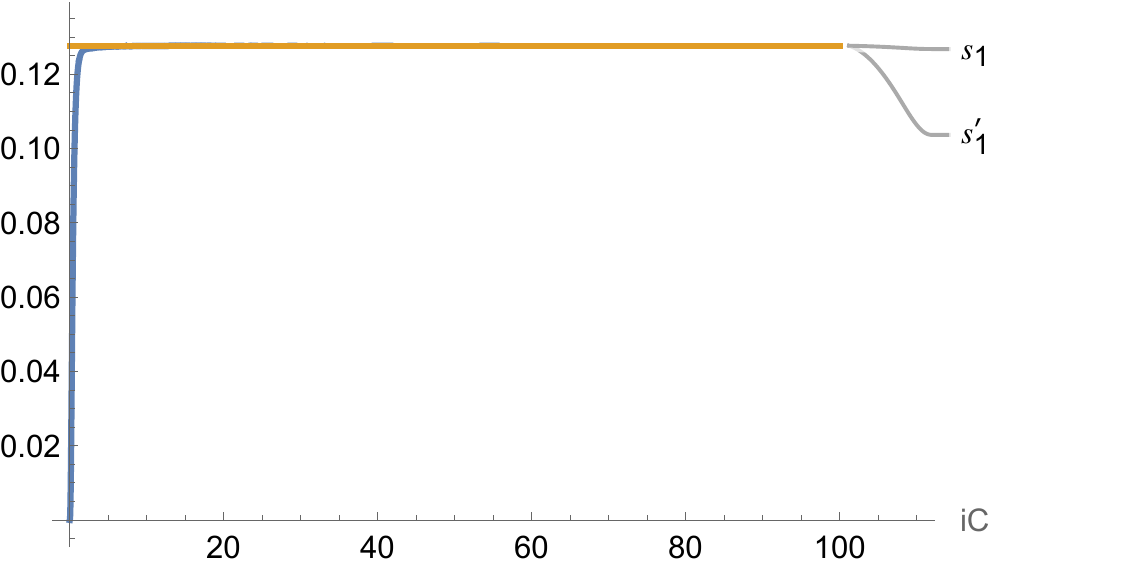}
    \includegraphics[width=0.45\linewidth]{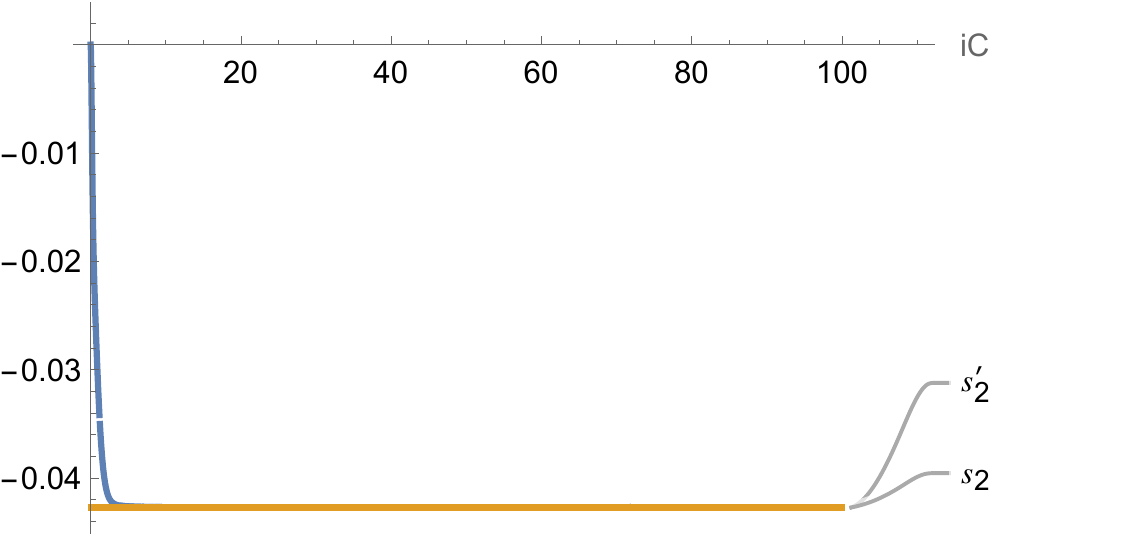}
    \includegraphics[width=0.45\linewidth]{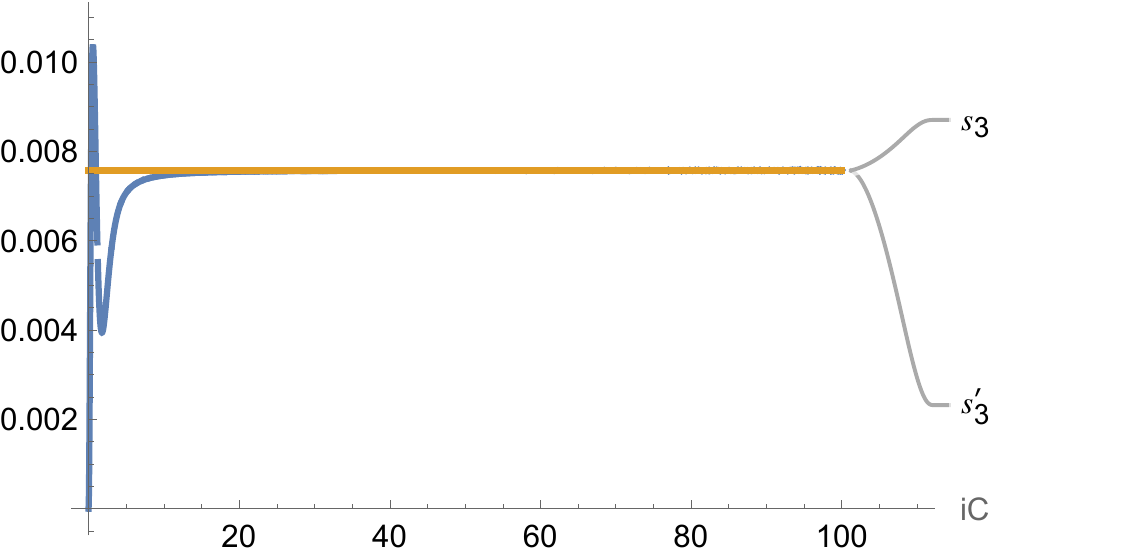}
    \caption{We compare $s_n$ and $s_n^\prime$ at $z=1/2$ for pure negative imaginary $C$, where $iC$ is pure real positive. }
    \label{fig:sn-plot}
\end{figure}
\end{remark}

\section{Alien culculus on $\Psi(Q_-e^{-\zeta})$ and $\Phi(Q_+e^{-\zeta})$}\label{sectionalien}

What we have shown are the Borel-Laplace expressions of $I_1$ and $I_2$ for $z$ near $\frac12$. In this section, we discuss the alien calculus of the Borel germs used in Proposition \ref{propositionI2BL} and Proposition \ref{propositionI1BL}. For abbreviation, we use
\begin{equation}\label{equationhphihpsi}
  \hpsi(\zeta):=\Psi(Q_-e^{-\zeta}) \quad \text{and} \quad   \hphi(\zeta):=\Phi(Q_+e^{-\zeta}).
\end{equation}
Of course, $\hpsi$ and $\hphi$ are functions of both $\zeta$ and $z$. We will later emphasize the dependence on $z$ and use the notations $\hpsi(\zeta, z)$ 
and $\hphi(\zeta, z)$ in Section \ref{subsection1-z}.

\begin{lemma}
    $\hpsi$ has singular points at
\begin{equation}\label{equationsingularpointsI2zeta}
    \{ \zeta^-_m := 2\pi i m , m\in\mathbb{Z}\}
    \quad\text{and}\quad\{ \zeta^+_m := -\log\left(\frac{Q_+}{Q_-}\right)+2\pi i m , m\in\mathbb{Z}\}
\end{equation}
and they are all integrable and double-branch. Similarly, $\hphi$ has singular points at 
\begin{equation}\label{equationsingularpointsI1zeta}
    \{ \zeta^-_m = 2\pi i m , m\in\mathbb{Z}\}
    \quad\text{and}\quad\{ -\zeta^+_{-m} = \log\left(\frac{Q_+}{Q_-}\right)+2\pi i m , m\in\mathbb{Z}\}
\end{equation}
and they are all integrable and double-branch.
\end{lemma}

\begin{proof}
Both $\Psi$ and $\Phi$ may have singularities at $Q_{\pm}$, as established in the discussion following equation \eqref{equationsingularsetsinducesbyQ}. Consequently, the locations of the singular points given in equations \eqref{equationsingularpointsI2zeta} and \eqref{equationsingularpointsI1zeta} are direct computations, visualized in Figure \ref{figuresingularsethpsi}. 

From equations \eqref{equationanalyticityofwz} and \eqref{equationw0w1holomorphic}, we deduce that $\Psi$ and $\Phi$ satisfy the monodromy relations in the $Q$-plane:
\begin{equation}
(\text{cont}_{Q_{\pm}})^2 \Psi = \Psi \quad \text{and} \quad (\text{cont}_{Q_{\pm}})^2 \Phi = \Phi.
\end{equation}
This confirms that $\widehat{\psi}$ and $\widehat{\phi}$ exhibit exclusively double-branch singularities.

Furthermore, by mimicking the analysis of forms $\frac{1}{(w_\alpha-w_\beta)(w_\alpha-w_\gamma)}$ (where $\{\alpha,\beta,\gamma\}=\{0,1,2\}$) presented in Lemma \ref{lemmaHpropoty}, we demonstrate the integrability of all singular points $\zeta_m^-$ and $\zeta_m^+$.
\end{proof}

For an integrable singularity at $\omega$ in direction $d$, the alien operator $\Delta_\omega^+$ is defined as the composition of an analytic continuation difference and a transition map:\footnote{%
    Readers familiar with alien calculus will recognize this as a restriction of the general theory—defined on spaces of singularities (cf \cite{Ecalle81}, Chap.~1 or \cite{S05}, section $3.3$)—to the case of integrable singularities.
}
\begin{equation}
    (\Delta_\omega^+ f)(\zeta) := \Big( \mathrm{cont}_{+} f - \mathrm{cont}_{\omega}  \mathrm{cont}_{+} f \Big)(\zeta + \omega),
\end{equation}
where $\mathrm{cont}_{+}$ denotes analytic continuation along a path $\gamma_{+}$ from a neighborhood of the origin to a neighborhood of $\omega$, avoiding singular points to the right. See Figure \ref{figureDelta+}.

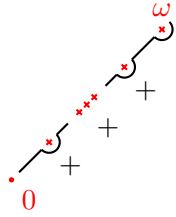
\begin{figure}[h]
\centering
    \begin{tikzpicture}[scale=0.5] 
    
    \foreach \point/\pos in {(1,1)/right,(3,3)/right,(4,4)/right,(2,2)/right,(2.2,2.2)/right,(1.8,1.8)/right}
    \draw[red,line width=0.8pt] \point ++(-2pt,-2pt) -- ++(4pt,4pt) ++(-4pt,0) -- ++(4pt,-4pt);
  \draw[thick] (0.2,0.2) -- (0.8,0.8); 
  \draw[thick] (1.2,1.2) -- (1.5,1.5); 
  \draw[thick] (2.5,2.5) -- (2.8,2.8); 
  \draw[thick] (3.2,3.2) -- (3.8,3.8); 

  \draw[thick] (1.2,1.2) arc (45:-135:0.282);
  \draw[thick] (3.2,3.2) arc (45:-135:0.282);
  \draw[thick] (4.2,4.2) arc (45:-135:0.282);
    \coordinate (A) at (0,0); 
  \fill[red] (A) circle (2pt) node[below right] {$0$};
  \node[below right] at (1,0.9) {$+$};
  \node[below right] at (2,1.9) {$+$};
  \node[below right] at (3,2.9) {$+$};
  \node[above,red] at (4,4.1) {$\omega$};
\end{tikzpicture}
\caption{For a point $\omega$ locaeed in direction $d$, the path used in ${\rm cont}_+$ is the black curve from a point near $0$ to a point near $\omega$. $+$ means going from right.}
\label{figureDelta+}
\end{figure}

Thus $\Delta^+$ is an operator that maps a germ at $\varepsilon e^{i d}$ (for $\varepsilon \ll 1$) to a germ holomorphic at the same base point. To conclude this brief introduction to the operator \(\Delta^+_\omega\), we emphasize---though we will not use this property in the present paper---that it satisfies a generalized Leibniz rule with respect to the convolution product.
\begin{equation}\label{equationalienleibniz}
    \Delta^+_{\omega} (f\ast g)= \sum\limits_{\omega_1+\omega_2=\omega \atop \omega_1,\omega_2\in [0,|\omega|]\cdot e^{id}} \Delta_{\omega_1}^+ f \ast \Delta_{\omega_2}^+ g, 
\end{equation}
for $f,g$ resurgent functions with only integrable singularities in direction $d$.
\ \\

On the other hand, one may similarly define $\Delta^-_\omega$. The relation between $\Delta^+$ and $\Delta^-$ will be mentioned in the Section \ref{subsectionStokesI2}.

\begin{lemma}\label{lemmaalienhphihpsi}
    For $\widehat{\psi}$ and $\hphi$ defined in \eqref{equationhphihpsi}, the alien operator acts as follows:
    \begin{equation}\label{equationalienhphihpsi}
        \Delta_{\zeta_m^-}^+ \widehat{\psi} = 2\widehat{\psi}, \quad 
        \Delta_{\zeta_m^+}^+ \widehat{\psi} = \hphi, \quad 
        \Delta_{\zeta_m^-}^+ \hphi= 2\hphi, \quad 
        \Delta_{-\zeta_{-m}^+}^+ \hphi = \widehat{\psi}.
    \end{equation}
\end{lemma}

\begin{proof}
    We compute the action on $\widehat{\psi}$, omitting the analogous case for $\hphi$. Using \eqref{equationanalyticityofwz} and \eqref{equationw0w1holomorphic}, we derive:
    \begin{align}
    \Delta_{\zeta_m^+}^+ \widehat{\psi}
    &= \Big[ \Psi(Q_-e^{-\zeta}) - \mathrm{cont}_{\zeta_m^+} \Psi (Q_-e^{-\zeta}) \Big]_{\zeta \to \zeta + \zeta_m^+} 
    \nonumber \\
    &= \left. Q \cdot \left( \frac{2w_1 - w_0 - w_2}{(w_0 - w_1)(w_0 - w_2)(w_1 - w_2)} 
    - \frac{2w_0 - w_1 - w_2}{(w_1 - w_0)(w_1 - w_2)(w_0 - w_2)} \right) \right|_{Q=Q_-e^{-\zeta - \zeta_m^+}}
    \nonumber \\
    &= \Phi(Q_+ e^{-\zeta}) = \hphi(\zeta),
    \end{align}
    and similarly for $\zeta_m^-$:
    \begin{equation}
        \Delta_{\zeta_m^-}^+ \widehat{\psi} 
        = \Big[ \Psi(Q_-e^{-\zeta}) - \mathrm{cont}_{\zeta_m^-} \Psi (Q_-e^{-\zeta}) \Big]_{\zeta \to \zeta + \zeta_m^-} 
        = 2\widehat{\psi}(\zeta).
    \end{equation}
\end{proof}

\subsection{The Stokes phenomenon}\label{subsectionStokesI2}

In this section, we investigate the Stokes phenomenon associated with the integral $I_2$. The central result establishes that $I_1$ emerges intrinsically as a Stokes correction term in the asymptotic expansion of $I_2$, and vice versa. Our approach employs the symbolic Stokes automorphism 
\footnote{By equation \eqref{equationalienleibniz}, the operator $\SA^+_d$ is a morphism of a graded algebra. For the details, see for instance \cite{Ecalle81} or \cite{CS16}.} 
$\SA^+_d$, which operates on the Borel resummation of $I_2$ to quantify its discontinuity across Stokes lines. Recall that the definition is as follows.
    \begin{equation}\label{equationdefinitionSA}
        \SA_d^+ := \text{Id} +\sum\limits_{\omega\in\Omega\cap d} e^{-C\omega}\Delta^+_\omega,
    \end{equation}
which satisfies
\begin{equation}
    \mathfrak{L}^{d-\varepsilon} =  \mathfrak{L}^{d+\varepsilon} \circ \SA^+_d.
\end{equation}
In equation \eqref{equationdefinitionSA}, $\Omega$ denotes the singular set of the target holomorphic function, and $d$ is a singular ray originating from the origin. For instance, when $z = \frac{1}{2}$, the set $\{\zeta_m^-, m \in \mathbb{Z}\} \cup \{\zeta_m^+, m \in \mathbb{Z}\} = i\pi\mathbb{Z}$ forms a lattice on the imaginary axis. Thus, for $z$ in a neighborhood of $\frac{1}{2}$, the singular sets of $\widehat{\psi}$ and $\hphi$ are \underline{almost} the lattice $\pi i\mathbb{Z}$. We will employ the symbolic Stokes automorphism ${\SA}^+_{\frac{\pi}{2}}$ to aggregate contributions from all singular points along the ray with direction \underline{almost} $\theta = \frac{\pi}{2}$. The geometric configuration is illustrated in Figure \ref{figuresingularsethpsi}. 

We recall that the term $e^{-C\omega}$ commutes as a symbol with the Laplace transform operator under this framework, i.e., 
\begin{equation}
    \mathfrak{L}^\theta \left( e^{-C\omega}  f \right) = e^{-C\omega}  \mathfrak{L}^\theta f.
\end{equation}

The operator $\SA^+_d$ performs from the right-hand to the left-hand Laplace transforms, while its symmetric dual $\SA^-_d$ governs the Stokes automorphism in the reverse direction.
    
We have that   \begin{equation}\label{equationdefinitionSA}
        \SA_d^- := \text{Id} +\sum\limits_{\omega\in\Omega\cap d} e^{-C\omega}\Delta^-_\omega
    \end{equation}
and it satisfies that 
\begin{equation}
    \mathfrak{L}^{d+\varepsilon} =  \mathfrak{L}^{d-\varepsilon} \circ \SA^-_d
\end{equation}
and 
\begin{equation}\label{equationSAid}
    \SA^-_d \circ \SA^+_d = \text{Id}.
\end{equation}

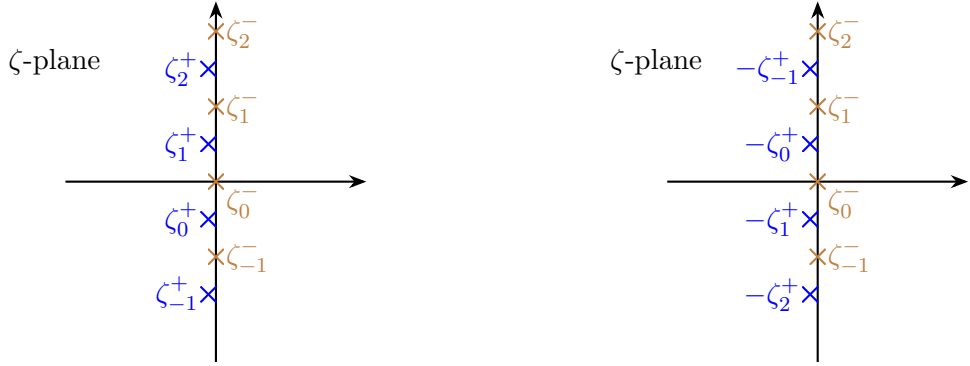
\begin{figure}[h]
\centering
\begin{tikzpicture}[scale=2,        
mid arrow/.style={
        postaction={decorate}, 
        decoration={
            markings,
            mark=at position 0.5 with {\arrow{Stealth[scale=1.2]}} 
        }
    }]

\begin{scope}[xshift=-2cm, local bounding box=left]
  \draw[-{Stealth[scale=1]}, thick] (-1,0)--(1,0); 
  \draw[-{Stealth[scale=1]}, thick] (0,-1.2)--(0,1.2); 
  \node[inner sep=3pt, path picture={
    \draw[thick,blue] (path picture bounding box.south west) -- (path picture bounding box.north east);
    \draw[thick,blue] (path picture bounding box.north west) -- (path picture bounding box.south east);
  }] at (-0.05,-0.25) {};
    \node[inner sep=3pt, path picture={
    \draw[thick,blue] (path picture bounding box.south west) -- (path picture bounding box.north east);
    \draw[thick,blue] (path picture bounding box.north west) -- (path picture bounding box.south east);
  }] at (-0.05,-0.75) {};
    \node[inner sep=3pt, path picture={
    \draw[thick,blue] (path picture bounding box.south west) -- (path picture bounding box.north east);
    \draw[thick,blue] (path picture bounding box.north west) -- (path picture bounding box.south east);
  }] at (-0.05,0.75) {};
  \node[inner sep=3pt, path picture={
    \draw[thick,blue] (path picture bounding box.south west) -- (path picture bounding box.north east);
    \draw[thick,blue] (path picture bounding box.north west) -- (path picture bounding box.south east);
  }] at (-0.05,0.25) {};
    \node[inner sep=3pt, path picture={
    \draw[thick,brown] (path picture bounding box.south west) -- (path picture bounding box.north east);
    \draw[thick,brown] (path picture bounding box.north west) -- (path picture bounding box.south east);
  }] at (0,0) {};
      \node[inner sep=3pt, path picture={
    \draw[thick,brown] (path picture bounding box.south west) -- (path picture bounding box.north east);
    \draw[thick,brown] (path picture bounding box.north west) -- (path picture bounding box.south east);
  }] at (0,0.5) {};
      \node[inner sep=3pt, path picture={
    \draw[thick,brown] (path picture bounding box.south west) -- (path picture bounding box.north east);
    \draw[thick,brown] (path picture bounding box.north west) -- (path picture bounding box.south east);
  }] at (0,1) {};
      \node[inner sep=3pt, path picture={
    \draw[thick,brown] (path picture bounding box.south west) -- (path picture bounding box.north east);
    \draw[thick,brown] (path picture bounding box.north west) -- (path picture bounding box.south east);
  }] at (0,-0.5) {};
  \node[left] at (-0.7,0.8) {$\zeta$-plane};
  \node[left,blue] at (-0.05,0.75) {$\zeta_2^+$};
  \node[left,blue] at (-0.05,0.25) {$\zeta_1^+$};
  \node[left,blue] at (-0.05,-0.25) {$\zeta_0^+$};
  \node[left,blue] at (-0.05,-0.75) {$\zeta_{-1}^+$};
   \node[right,brown] at (0,-0.5) {$\zeta_{-1}^-$};
   \node[right,brown] at (0,0.5) {$\zeta_{1}^-$};
   \node[right,brown] at (0,1) {$\zeta_{2}^-$};
   \node[below right,brown] at (0,0.05) {$\zeta_{0}^-$};
\end{scope}

   \begin{scope}[xshift=2cm,scale=1]
\draw[-{Stealth[scale=1]}, thick] (-1,0)--(1,0); 
  \draw[-{Stealth[scale=1]}, thick] (0,-1.2)--(0,1.2); 
  \node[inner sep=3pt, path picture={
    \draw[thick,blue] (path picture bounding box.south west) -- (path picture bounding box.north east);
    \draw[thick,blue] (path picture bounding box.north west) -- (path picture bounding box.south east);
  }] at (-0.05,-0.25) {};
    \node[inner sep=3pt, path picture={
    \draw[thick,blue] (path picture bounding box.south west) -- (path picture bounding box.north east);
    \draw[thick,blue] (path picture bounding box.north west) -- (path picture bounding box.south east);
  }] at (-0.05,-0.75) {};
    \node[inner sep=3pt, path picture={
    \draw[thick,blue] (path picture bounding box.south west) -- (path picture bounding box.north east);
    \draw[thick,blue] (path picture bounding box.north west) -- (path picture bounding box.south east);
  }] at (-0.05,0.75) {};
  \node[inner sep=3pt, path picture={
    \draw[thick,blue] (path picture bounding box.south west) -- (path picture bounding box.north east);
    \draw[thick,blue] (path picture bounding box.north west) -- (path picture bounding box.south east);
  }] at (-0.05,0.25) {};
    \node[inner sep=3pt, path picture={
    \draw[thick,brown] (path picture bounding box.south west) -- (path picture bounding box.north east);
    \draw[thick,brown] (path picture bounding box.north west) -- (path picture bounding box.south east);
  }] at (0,0) {};
      \node[inner sep=3pt, path picture={
    \draw[thick,brown] (path picture bounding box.south west) -- (path picture bounding box.north east);
    \draw[thick,brown] (path picture bounding box.north west) -- (path picture bounding box.south east);
  }] at (0,0.5) {};
      \node[inner sep=3pt, path picture={
    \draw[thick,brown] (path picture bounding box.south west) -- (path picture bounding box.north east);
    \draw[thick,brown] (path picture bounding box.north west) -- (path picture bounding box.south east);
  }] at (0,1) {};
      \node[inner sep=3pt, path picture={
    \draw[thick,brown] (path picture bounding box.south west) -- (path picture bounding box.north east);
    \draw[thick,brown] (path picture bounding box.north west) -- (path picture bounding box.south east);
  }] at (0,-0.5) {};
  \node[left] at (-0.7,0.8) {$\zeta$-plane};
  \node[left,blue] at (-0.05,0.75) {$-\zeta_{-1}^+$};
  \node[left,blue] at (-0.05,0.25) {$-\zeta_0^+$};
  \node[left,blue] at (-0.05,-0.25) {$-\zeta_1^+$};
  \node[left,blue] at (-0.05,-0.75) {$-\zeta_{2}^+$};
   \node[right,brown] at (0,-0.5) {$\zeta_{-1}^-$};
   \node[right,brown] at (0,0.5) {$\zeta_{1}^-$};
   \node[right,brown] at (0,1) {$\zeta_{2}^-$};
   \node[below right,brown] at (0,0.05) {$\zeta_{0}^-$};
  \end{scope}
\end{tikzpicture}
\caption{When $z$ lies in a neighborhood of $\frac{1}{2}$, the singular set of $\hpsi$ is depicted in the left panel, while that of $\hphi$ appears in the right panel. They are both lattice $\pi i\mathbb{Z}$ as $z=\frac12$.
}
\label{figuresingularsethpsi}
\end{figure}

\begin{lemma}
       Recall $\hpsi$ in equation \eqref{equationhphihpsi}. As $z$ near $\frac12$, we have
       \begin{equation}\label{equationSAhpsi}
           \SA^+_{\frac\pi2}\hpsi = \hpsi + 2\sum\limits_{m\geq1} e^{-2\pi imC} \hpsi + \left(\frac{Q_+}{Q_-}\right)^C \sum\limits_{m\geq1} e^{-2\pi imC} \hphi.
       \end{equation}
\end{lemma}

\begin{proof}
  This is a result of the definition of $\SA_{\frac\pi2}^+$ in equation \eqref{equationdefinitionSA} and alien calculus in Lemma \ref{lemmaalienhphihpsi}.
\end{proof}

\begin{lemma}\label{lemmaStokeshpsi}
Stokes phenomenon of $\hpsi$ is
    \begin{equation}
        \mathfrak{L}^{0} \hpsi = \mathfrak{L}^{\frac{\pi}{2}+\varepsilon} \circ \SA^+_{\frac\pi2} \hpsi =
        \mathfrak{L}^{\frac{\pi}{2}+\varepsilon}\hpsi + 2\sum\limits_{m\geq1} e^{-2\pi imC} \mathfrak{L}^{\frac{\pi}{2}+\varepsilon}\hpsi + \left(\frac{Q_+}{Q_-}\right)^C \sum\limits_{m\geq1} e^{-2\pi imC} \mathfrak{L}^{\frac{\pi}{2}+\varepsilon}\hphi.
    \end{equation}
\end{lemma}

\begin{proof}
    Use equations \eqref{equationdefinitionSA} and \eqref{equationSAhpsi}. 
\end{proof}

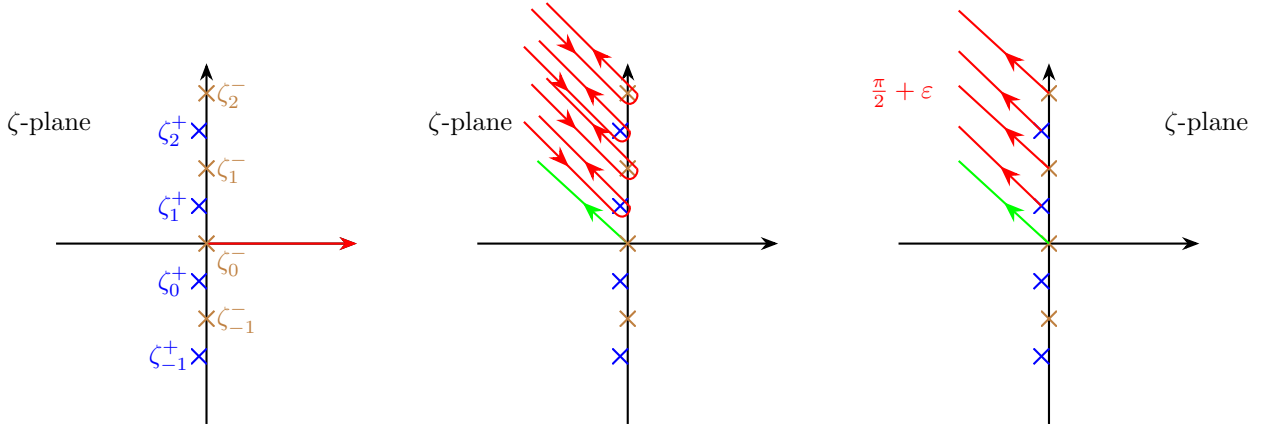
\begin{figure}[h]
\centering
\begin{tikzpicture}[scale=2, mid arrow/.style={
        postaction={decorate}, 
        decoration={
            markings,
            mark=at position 0.5 with {\arrow{Stealth[scale=1.2]}} 
        }
    },font=\small]

\begin{scope}[xshift=-1cm, local bounding box=left]
  \draw[-{Stealth[scale=1]}, thick] (-1,0)--(1,0); 
  \draw[-{Stealth[scale=1]}, thick,red] (0,0)--(1,0);
  \draw[-{Stealth[scale=1]}, thick] (0,-1.2)--(0,1.2); 
  \node[inner sep=3pt, path picture={
    \draw[thick,blue] (path picture bounding box.south west) -- (path picture bounding box.north east);
    \draw[thick,blue] (path picture bounding box.north west) -- (path picture bounding box.south east);
  }] at (-0.05,-0.25) {};
    \node[inner sep=3pt, path picture={
    \draw[thick,blue] (path picture bounding box.south west) -- (path picture bounding box.north east);
    \draw[thick,blue] (path picture bounding box.north west) -- (path picture bounding box.south east);
  }] at (-0.05,-0.75) {};
    \node[inner sep=3pt, path picture={
    \draw[thick,blue] (path picture bounding box.south west) -- (path picture bounding box.north east);
    \draw[thick,blue] (path picture bounding box.north west) -- (path picture bounding box.south east);
  }] at (-0.05,0.75) {};
  \node[inner sep=3pt, path picture={
    \draw[thick,blue] (path picture bounding box.south west) -- (path picture bounding box.north east);
    \draw[thick,blue] (path picture bounding box.north west) -- (path picture bounding box.south east);
  }] at (-0.05,0.25) {};
    \node[inner sep=3pt, path picture={
    \draw[thick,brown] (path picture bounding box.south west) -- (path picture bounding box.north east);
    \draw[thick,brown] (path picture bounding box.north west) -- (path picture bounding box.south east);
  }] at (0,0) {};
      \node[inner sep=3pt, path picture={
    \draw[thick,brown] (path picture bounding box.south west) -- (path picture bounding box.north east);
    \draw[thick,brown] (path picture bounding box.north west) -- (path picture bounding box.south east);
  }] at (0,0.5) {};
      \node[inner sep=3pt, path picture={
    \draw[thick,brown] (path picture bounding box.south west) -- (path picture bounding box.north east);
    \draw[thick,brown] (path picture bounding box.north west) -- (path picture bounding box.south east);
  }] at (0,1) {};
      \node[inner sep=3pt, path picture={
    \draw[thick,brown] (path picture bounding box.south west) -- (path picture bounding box.north east);
    \draw[thick,brown] (path picture bounding box.north west) -- (path picture bounding box.south east);
  }] at (0,-0.5) {};
  \node[left] at (-0.7,0.8) {$\zeta$-plane};
  \node[left,blue] at (-0.05,0.75) {$\zeta_2^+$};
  \node[left,blue] at (-0.05,0.25) {$\zeta_1^+$};
  \node[left,blue] at (-0.05,-0.25) {$\zeta_0^+$};
  \node[left,blue] at (-0.05,-0.75) {$\zeta_{-1}^+$};
   \node[right,brown] at (0,-0.5) {$\zeta_{-1}^-$};
   \node[right,brown] at (0,0.5) {$\zeta_{1}^-$};
   \node[right,brown] at (0,1) {$\zeta_{2}^-$};
   \node[below right,brown] at (0,0.05) {$\zeta_{0}^-$};
\end{scope}

\begin{scope}[xshift=1.8cm, local bounding box=middle]
  \draw[-{Stealth[scale=1]}, thick] (-1,0)--(1,0); 
  \draw[thick,green,mid arrow] (0,0)--(-0.6,0.55); 
  \draw[-{Stealth[scale=1]}, thick] (0,-1.2)--(0,1.2); 
  \node[inner sep=3pt, path picture={
    \draw[thick,blue] (path picture bounding box.south west) -- (path picture bounding box.north east);
    \draw[thick,blue] (path picture bounding box.north west) -- (path picture bounding box.south east);
  }] at (-0.05,-0.25) {};
    \node[inner sep=3pt, path picture={
    \draw[thick,blue] (path picture bounding box.south west) -- (path picture bounding box.north east);
    \draw[thick,blue] (path picture bounding box.north west) -- (path picture bounding box.south east);
  }] at (-0.05,-0.75) {};
    \node[inner sep=3pt, path picture={
    \draw[thick,blue] (path picture bounding box.south west) -- (path picture bounding box.north east);
    \draw[thick,blue] (path picture bounding box.north west) -- (path picture bounding box.south east);
  }] at (-0.05,0.75) {};
  \node[inner sep=3pt, path picture={
    \draw[thick,blue] (path picture bounding box.south west) -- (path picture bounding box.north east);
    \draw[thick,blue] (path picture bounding box.north west) -- (path picture bounding box.south east);
  }] at (-0.05,0.25) {};
    \node[inner sep=3pt, path picture={
    \draw[thick,brown] (path picture bounding box.south west) -- (path picture bounding box.north east);
    \draw[thick,brown] (path picture bounding box.north west) -- (path picture bounding box.south east);
  }] at (0,0) {};
      \node[inner sep=3pt, path picture={
    \draw[thick,brown] (path picture bounding box.south west) -- (path picture bounding box.north east);
    \draw[thick,brown] (path picture bounding box.north west) -- (path picture bounding box.south east);
  }] at (0,0.5) {};
      \node[inner sep=3pt, path picture={
    \draw[thick,brown] (path picture bounding box.south west) -- (path picture bounding box.north east);
    \draw[thick,brown] (path picture bounding box.north west) -- (path picture bounding box.south east);
  }] at (0,1) {};
      \node[inner sep=3pt, path picture={
    \draw[thick,brown] (path picture bounding box.south west) -- (path picture bounding box.north east);
    \draw[thick,brown] (path picture bounding box.north west) -- (path picture bounding box.south east);
  }] at (0,-0.5) {};
  \node[left] at (-0.7,0.8) {$\zeta$-plane};

  

  \draw[thick, red, mid arrow](-0.64,1.06)--(-0.04,0.46);
  \draw[thick, red, mid arrow](0.06,0.5)--(-0.54,1.1);
\draw[thick,red] (-0.04,0.46) .. controls (0.03,0.38) and (0.08,0.48) .. (0.06,0.5);

  \draw[thick, red, mid arrow](-0.64,1.56)--(-0.04,0.96);
  \draw[thick, red, mid arrow](0.06,1)--(-0.54,1.6);
\draw[thick,red] (-0.04,0.96) .. controls (0.03,0.88) and (0.08,0.98) .. (0.06,1);

  \draw[thick, red, mid arrow](-0.69,1.31)--(-0.09,0.71);
  \draw[thick, red, mid arrow](0.01,0.75)--(-0.59,1.35);
\draw[thick,red] (-0.09,0.71) .. controls (-0.02,0.63) and (0.03,0.73) .. (0.01,0.75);

  \draw[thick, red, mid arrow](-0.69,0.81)--(-0.09,0.21);
  \draw[thick, red, mid arrow](0.01,0.25)--(-0.59,0.85);
\draw[thick,red] (-0.09,0.21) .. controls (-0.02,0.13) and (0.03,0.23) .. (0.01,0.25);
\end{scope}

\begin{scope}[xshift=4.6cm, local bounding box=right]
  \draw[-{Stealth[scale=1]}, thick] (-1,0)--(1,0); 
  \draw[-{Stealth[scale=1]}, thick] (0,-1.2)--(0,1.2); 
  \node[inner sep=3pt, path picture={
    \draw[thick,blue] (path picture bounding box.south west) -- (path picture bounding box.north east);
    \draw[thick,blue] (path picture bounding box.north west) -- (path picture bounding box.south east);
  }] at (-0.05,-0.25) {};
    \node[inner sep=3pt, path picture={
    \draw[thick,blue] (path picture bounding box.south west) -- (path picture bounding box.north east);
    \draw[thick,blue] (path picture bounding box.north west) -- (path picture bounding box.south east);
  }] at (-0.05,-0.75) {};
    \node[inner sep=3pt, path picture={
    \draw[thick,blue] (path picture bounding box.south west) -- (path picture bounding box.north east);
    \draw[thick,blue] (path picture bounding box.north west) -- (path picture bounding box.south east);
  }] at (-0.05,0.75) {};
  \node[inner sep=3pt, path picture={
    \draw[thick,blue] (path picture bounding box.south west) -- (path picture bounding box.north east);
    \draw[thick,blue] (path picture bounding box.north west) -- (path picture bounding box.south east);
  }] at (-0.05,0.25) {};
    \node[inner sep=3pt, path picture={
    \draw[thick,brown] (path picture bounding box.south west) -- (path picture bounding box.north east);
    \draw[thick,brown] (path picture bounding box.north west) -- (path picture bounding box.south east);
  }] at (0,0) {};
      \node[inner sep=3pt, path picture={
    \draw[thick,brown] (path picture bounding box.south west) -- (path picture bounding box.north east);
    \draw[thick,brown] (path picture bounding box.north west) -- (path picture bounding box.south east);
  }] at (0,0.5) {};
      \node[inner sep=3pt, path picture={
    \draw[thick,brown] (path picture bounding box.south west) -- (path picture bounding box.north east);
    \draw[thick,brown] (path picture bounding box.north west) -- (path picture bounding box.south east);
  }] at (0,1) {};
      \node[inner sep=3pt, path picture={
    \draw[thick,brown] (path picture bounding box.south west) -- (path picture bounding box.north east);
    \draw[thick,brown] (path picture bounding box.north west) -- (path picture bounding box.south east);
  }] at (0,-0.5) {};
  \node[right] at (0.7,0.8) {$\zeta$-plane};
  \draw[thick,green,mid arrow] (0,0)--(-0.6,0.55); 
  \draw[thick,red,mid arrow] (0,0.5)--(-0.6,1.05); 
  \draw[thick,red,mid arrow] (0,1)--(-0.6,1.55); 
  \draw[thick,red,mid arrow] (-0.05,0.25)--(-0.6,0.78); 
  \draw[thick,red,mid arrow] (-0.05,0.75)--(-0.6,1.28); 
  \node[left,red] at (-0.7,1) {$\frac\pi2+\varepsilon$};
\end{scope}
\end{tikzpicture}
\caption{
When $z$ is near $\frac12$, as in Proposition \ref{propositionI2BL}, $\hpsi$ as singular points shown in the left picture. The middle picture is the integral curve that we use when we perform analytic continuation in the variable $C$ from $\mathfrak{Re}C>0$ to $\mathfrak{Im}C<0$. These integrals (except for the green one) provide the Stokes terms along the (almost) $\frac\pi2$ direction. Since every singular point is integrable, these integrals on the Hankel contours can be expressed by the Laplace transform on $\Delta_\omega^+\hpsi$ along the $\frac\pi2+\varepsilon$ direction attached to each singular point $\omega$, as shown in the right picture. In the proof of Lemma \ref{lemmaStokeshpsi}, we have to move every integral back to the origin such that they are the standard Laplace transform starting at the origin.}
\label{figureStokesI2}
\end{figure}

Thus, those alien computations yield our main theorem.

\begin{theorem}\label{theoremI2Stokes}
    Recall that in Proposition \ref{propositionI2BL} and Proposition \ref{propositionI1BL}:
    \begin{equation}
        I_1(C,z) = Q_+^C \frac{1}{e^{2\pi iC}-1}\mathfrak{L}^{\frac\pi2+\varepsilon} \hphi ,\quad I_2=Q_-^{C} \mathfrak{L}^0 \hpsi
    \end{equation}
as $z$ near $\frac12$. 
The Stokes phenomenon of $I_2$ associated with these singular points on the positive imaginary half plane is
    \begin{equation}\label{equationI2Stokes}
        I_2(C,z)=Q_-^C \frac{1+e^{-2\pi iC}}{1-e^{-2\pi iC}}\mathfrak{L}^{\frac\pi2+\varepsilon} \Psi(Q_-e^{-\zeta}) +I_1(C,z)  , \quad \mathfrak{Im} C<0.
    \end{equation}
The trans-series expansion of $I_2$ as $C$ goes to $\infty$ with $\mathfrak{Im} C<0$ is
\begin{equation}
\begin{split}
     I_2(C,z) 
     &\sim
     Q_-^{C} \mathfrak{B}^{-1} \Psi(Q_-e^{-\zeta})
    \\
    &+
    2Q_-^C\sum\limits_{m\geq0} e^{-2\pi imC} \mathfrak{B}^{-1} \Psi(Q_-e^{-\zeta})
    +
     Q_+^C \sum\limits_{m\geq1} e^{-2\pi imC} \mathfrak{B}^{-1} \Phi(Q_+e^{-\zeta}).
\end{split}
\end{equation}
\end{theorem}

\begin{proof}
This is a direct result of Lemma \ref{lemmaStokeshpsi}.
\end{proof}

\begin{remark}
One may also view equation \eqref{equationI2Stokes} as a formula for $C<0$ since both terms in the right-hand side are well-defined when $C<0$. On the other hand, we have compared the asymptotic expansion $\mathfrak{B}^{-1} \Psi(Q_-e^{-\zeta})$ with $\frac{1-e^{-2\pi iC}}{1+e^{-2\pi iC}}Q_-^{-C}(I_2-I_1)$ using the hypergeometric expression of $I_1$ \eqref{eq:I1-hyper} and $I_2$:
\begin{equation}
    I_{2}(C,z)=z^{1+2C}\frac{\Gamma(1+C)^{2}}{\Gamma(2+2C)}{}_{2}F_{1}(-C,C+1,2C+2,z).
\end{equation}
This can be done in a similar way as remark \ref{rmk:I1-asym} and shows a precise match.
\end{remark}

\begin{remark}
    Using the identity of hypergeometric function (see for instance eq.(12.1.12) in \cite{Temme}), we find
    \begin{equation}
    I_{2}(z)=I_{1}(z)-\frac{\sin(2\pi C)}{\sin(\pi C)}I_{1}(1-z).
\end{equation}
Comparing with \eqref{equationI2Stokes}, one thus finds
\begin{equation}\label{equationI11-z}
    -Q_{-}^{-C}\frac{1-e^{-2\pi iC}}{1+e^{-2\pi iC}}\frac{\sin(2\pi C)}{\sin(\pi C)}I_{1}(1-z)=\mathfrak{L}^{\frac{\pi}{2}+\varepsilon}\Psi(Q_{-}e^{-\zeta}),\quad \mathfrak{Im}C<0,
\end{equation}
which has been tested via the asymptotic expansion.
\end{remark}

On the other side, the Stokes phenomenon of $I_1$ (from left-hand to right-hand Laplace transform) is related to the Stokes automorphism $\SA^-_{\frac\pi2}$. We shall compute the alien operators $\Delta^-_{\bullet}$ first.
\begin{lemma}
    \begin{equation}\label{equationDelta-}
    \begin{split}
        &\Delta^-_{\zeta^-_m} \hpsi =
        \begin{cases}
         2\hpsi  \quad &\text{if } m=3n, \ (n\geq 1) \\
         -\hpsi \quad &\text{if else } (m \geq 1).
        \end{cases}
        \quad 
        \Delta^-_{\zeta^+_m} \hpsi =
        \begin{cases}
         2\hphi  \quad &\text{if } m=3n-1, \ (n\geq 1) \\
         -\hphi \quad &\text{if else } (m \geq 1).
        \end{cases}
        \\
        &\Delta^-_{\zeta^-_m} \hphi =
        \begin{cases}
         2\hphi  \quad &\text{if } m=3n, \ (n\geq 1) \\
         -\hphi \quad &\text{if else } (m \geq 1).
        \end{cases}
        \quad 
        \Delta^-_{-\zeta^+_{-m}} \hphi =
        \begin{cases}
         2\hpsi  \quad &\text{if } m=3n-2, \ (n\geq 1) \\
         -\hpsi \quad &\text{if else } (m \geq 0) .
         \end{cases}
    \end{split}
    \end{equation}
\end{lemma}
\begin{proof}
    By Lemma \ref{lemmaalienhphihpsi}, we have
\begin{equation}\label{equationSAhphi}
           \SA^+_{\frac\pi2}\hphi = \hphi + 2\sum\limits_{m\geq1} e^{-2\pi imC} \hphi + \left(\frac{Q_-}{Q_+}\right)^C \sum\limits_{m\geq 0} e^{-2\pi imC} \hpsi
       \end{equation}
       similar to \eqref{equationSAhpsi}. Together with \eqref{equationSAhpsi} and \eqref{equationSAid}, one can iteratively determine the alien operators in equation \eqref{equationDelta-}.
\end{proof}
We thus have the Stokes phenomenon of $I_1$ in direction $\frac\pi 2$.
By the above Lemma, we have
\begin{equation}
\begin{split}
\SA^-_{\frac\pi2}\hphi 
&=
\hphi + \left(2\sum\limits_{m\geq1 \atop m=3n}  e^{-2\pi imC} -\sum\limits_{m\geq1 \atop m \neq 3n}  e^{-2\pi imC} \right)\hphi + \left(\frac{Q_-}{Q_+}\right)^C \left(2\sum\limits_{m\geq 0 \atop m=3n-2} e^{-2\pi imC} - \sum\limits_{m\geq 0 \atop m\neq3n-2} e^{-2\pi 
imC} \right) \hpsi    
\\
&=\left( \frac{e^{6\pi iC}-e^{4\pi iC}-e^{2\pi iC}+1}{e^{6\pi iC}-1} \right)\hphi + \left(\frac{Q_-}{Q_+}\right)^C \left(\frac{-e^{6\pi iC}+2e^{4\pi iC}-e^{2\pi iC}}{e^{6\pi iC}-1}\right) \hpsi  
\\
&= \frac{\sin(2\pi C)}{\sin(3\pi C)} (e^{2\pi iC}-1)e^{-\pi iC}\hphi - \left(\frac{Q_-}{Q_+}\right)^C \frac{\sin(\pi C)}{\sin(3\pi C)} (e^{2\pi iC}-1) \hpsi
\end{split}
\end{equation}
which implies that 
\begin{equation}\label{equationStokesI1}
    I_1(C,z) = Q_+^C\frac{\sin(2\pi C)}{\sin(3\pi C)} e^{-\pi iC} \mathfrak{L}^0 \hphi - \frac{\sin(\pi C)}{\sin(3\pi C)}  I_2(C,z).
\end{equation}

\begin{remark}
    Using the identity of the hypergeometric function, we find
    \begin{equation}
    I_{1}(z)+\frac{\sin(\pi C)}{\sin(3\pi C)}I_{2}(z)=-\frac{\sin(2\pi C)}{\sin(3\pi C)}I_{2}(1-z).
\end{equation}\label{equationStokesI1inI1I2}
Comparing with \eqref{equationStokesI1}, one thus finds
\begin{equation}\label{equationI21-z}
   Q_{+}^{C}e^{-\pi iC}\mathfrak{L}^{0}\hat{\varphi}=-I_{2}(1-z),
\end{equation}
which has been tested via the asymptotic expansion.
\end{remark}

\eqref{equationI2Stokes} and \eqref{equationStokesI1} are the main results in our paper. This analysis demonstrates that the conformal block $I_1(C,z)$ appears in the Stokes phenomenon of the other conformal $I_2$ (and vice versa). Given a certain conformal block, resurgence theory enables us to discover other internal operators (conformal blocks) \cite{Benjamin:2023uib,Bissi:2024wur}. 

\subsection{Change $z$ to $1-z$}\label{subsection1-z}

We have established a duality between the Borel germs $\hphi$ and $\hpsi$ mediated by alien operators. From the perspective of resurgence theory, the substitution $z \mapsto 1-z$ constitutes another unexpected duality.

\begin{lemma}
    As $z$ around $\frac12$, we have
    \begin{equation}
        Q_{\pm}(1-z) = - Q_{\mp}(z), \quad u(1-z) =u(z), \quad 
        v(Q_{\pm}(1-z),1-z)=-v(Q_{\mp}(z),z),
    \end{equation}
    where $Q_{\pm} ,u,v$ are viewed as functions of $z$ ($v$ has extra variable $Q$) defined in \eqref{equationcriticalvalueQpm} and \eqref{equationuv}.
\end{lemma}
\begin{proof}
    By direct computation.
\end{proof}

\begin{proposition}\label{proposition1-z}
    Fixed $z$ around $\frac12$. We have
    \begin{equation}
        \hpsi(\zeta,1-z) = -\hphi(\zeta,z).
    \end{equation}
\end{proposition}

\begin{proof}
    By definition, we have
    \begin{equation}
        \hpsi(\zeta,1-z) = Q_{-}(1-z) e^{-\zeta} \left.\frac{2w_1-w_0-w_2}{(w_0-w_1)(w_0-w_2)(w_1-w_2)}\right|_{\big(Q_-(1-z)e^{-\zeta},1-z\big)}.
    \end{equation}
    The difference of $w_i's$ has a nice property. Indeed, by the Lemma above,
    \begin{equation}
        \begin{split}
            &(w_0-w_1)|_{\big(Q_-(1-z)e^{-\zeta},1-z\big)} 
            = 
            2\left(-\frac{u(z)}{3}\right)^{\frac12}\cos\left(\frac{1}{3}\arccos\left(-\frac{3v(Q_+(z),z)}{2u(z)}\left(-\frac{3}{u(z)}\right)^{\frac12}\right)+\frac{2\pi}{3}\right)
            \\
            -2&\left(-\frac{u(z)}{3}\right)^{\frac12}\cos\left(\frac{1}{3}\arccos\left(-\frac{3v(Q_+(z)e^{-\zeta},z)}{2u(z)}\left(-\frac{3}{u(z)}\right)^{\frac12}\right)\right)
            \\
            &=
            2\left(-\frac{u(z)}{3}\right)^{\frac12} \left.\left(-\cos\left(\frac{1}{3}\arccos\left(\frac{3v}{2u}\left(-\frac{3}{u}\right)^{\frac12}\right)\right)+\cos\left(\frac{1}{3}\arccos\left(\frac{3v}{2u}\left(-\frac{3}{u}\right)^{\frac12}\right)+\frac23\pi\right)\right)\right|_{(Q_+e^{-\zeta},z)}
            \\
            &=
            (w_0-w_1)|_{(Q_+e^{-\zeta},z)},
        \end{split}
    \end{equation}
    where formula \eqref{equationarccos-x} is used in the second step. Similarly, we have 
    \begin{equation}
    \begin{split}
        (w_0-w_2)|_{\big(Q_-(1-z)e^{-\zeta},1-z\big)} 
            &=
            (w_2-w_1)|_{(Q_+e^{-\zeta},z)}
            \\
            (w_1-w_2)|_{\big(Q_-(1-z)e^{-\zeta},1-z\big)} 
            &=
            (w_2-w_0)|_{(Q_+e^{-\zeta},z)} 
    \end{split}
    \end{equation}
    We get the desired result using these formulas for the difference of $w$'s.
\end{proof}

Now we have that 
\begin{equation}
    I_1(C,1-z) = -Q_{+}^C(1-z) \frac{1}{e^{2\pi iC}-1} \mathfrak{L}^{\frac\pi2+\varepsilon} \hpsi.
\end{equation}

Combining with equation \eqref{equationI11-z}, we get the canonical choice in our system
\begin{equation}
    Q_+^C(1-z) = e^{\pi iC}Q_-^C(z), \quad Q_-^C(1-z)=e^{-\pi iC}Q_+^C(z).
\end{equation}

We thus have the Borel-Laplace expression for $I_1(1-z)$ and $I_2(1-z)$.

\begin{proposition}\label{pro:Izto1-z}
    Let $z$ be around $\frac12$. We have
    \begin{equation}
        I_1(C,1-z) =  -Q_{-}^C(z) \frac{e^{\pi iC}}{e^{2\pi iC}-1} \mathfrak{L}^{\frac\pi2+\varepsilon} \hpsi, \quad 
        I_2(C,1-z) =  -Q_{+}^C(z) e^{-\pi iC}
        \mathfrak{L}^{0} \hphi.
    \end{equation}
\end{proposition}

One can recover formula \eqref{equationStokesI1inI1I2} by using the Stokes phenomenon of $I_1$ in \eqref{equationStokesI1} and the above equation. Since the four-point correlation function is supposed to be invariant under the transform $z\to 1-z$, which is also known as crossing symmetry, proposition \ref{pro:Izto1-z} is expected to play a key role in expressing the correlation function as a trans-series.


\section{Monodromy in $z$}\label{sectionmonodromy}

In this section, we discuss the monodromy behavior of $I_1$ and $I_2$ in variable $z$ at $0$ and $1$. We denote 
\begin{equation}
    {\rm cont}_{z=p} f(z) 
\end{equation}
to be the analytic continuation of the holomorphic function $f$ that goes negatively around a point $p$.\footnote{See the Appendix \ref{Appendixarccos}. Here we emphasize that the variable we are discussing is $z$.} Since we already have the Borel-Laplace formula for $I_1,I_2$:
\begin{equation}
    I_1(C,z) = Q_+^C \frac{1}{e^{2\pi iC}-1}\mathfrak{L}^{\frac\pi2+\varepsilon} \hphi, \quad I_2(C,z) = Q_-^C \mathfrak{L}^0\hpsi,
\end{equation}
the monodromy computation on such germs is highly related to the trace of the singular points of $\hphi$ and $\hpsi$ on the Borel plane as $z$ moves. We start with a lemma which can be directly proved.

\begin{lemma}\label{lemmamonodromyQ}
    Let $z$ be around $\frac12$ at the beginning. According to the formula \eqref{equationcriticalvalueQpm} of $Q_{\pm}$, we have
    \begin{equation}\label{equationmonodromyQ}
            {\rm cont}_{z=0} Q_-^C = e^{-4\pi iC} Q_-^C, \quad {\rm cont}_{z=0} Q_+^C = Q_+^C, \quad
            {\rm cont}_{z=1} Q_-^C =  Q_-^C, \quad {\rm cont}_{z=1}Q_+^C= e^{-4\pi iC}Q_+^C.
    \end{equation}
Moreover,
 \begin{equation}
{\rm cont}_{z=0}\log{\frac{Q_+}{Q_-}} = \log{\frac{Q_+}{Q_-}} +4\pi i, \quad {\rm cont}_{z=1}\log{\frac{Q_+}{Q_-}} = \log{\frac{Q_+}{Q_-}} -4\pi i.
    \end{equation}
    
\end{lemma}

One may notice that $\frac{Q_+}{Q_-}$ is real as $z\in(0,1)$. The module $|\frac{Q_+}{Q_-}|>\!\!>1$ as $z$ near $0$ and $|\frac{Q_+}{Q_-}|<\!\!<1$ as $z$ near $1$. Together with the result in Lemma {\ref{lemmamonodromyQ}, we can draw the following picture.

\begin{figure}[h]
    \centering
\begin{tikzpicture}[scale=1, every node/.style={font=\footnotesize},mid arrow/.style={
        postaction={decorate}, 
        decoration={
            markings,
            mark=at position 0.5 with {\arrow{Stealth[scale=1.2]}} 
        }
    }]

  \begin{scope}[shift={(-2.5,1.5)}]
    \fill[blue] (1,0) circle (2pt) node[right] {$\frac{1}{2}$}; 
    \fill[red] (0,0) circle (2pt) node[above left] {$0$}; 
     \draw[thick, blue, mid arrow](0.9,-0.05)--(0.1,-0.05);
  \draw[thick, blue, mid arrow](0.1,0.05)--(0.9,0.05);
\draw[thick,blue] (0.1,-0.05) .. controls (-0.25,-0.25) and (-0.25,0.25) .. (0.1,0.05);
 \node[left] at (-0.7,0.8) {$z$-plane};
  \end{scope}

  \begin{scope}[shift={(2.5,1.5)}]
    \draw[->] (-2,0) -- (2,0) node[right] {$\mathrm{Re}$}; 
    \draw[->] (0,-1) -- (0,3) node[above] {$\mathrm{Im}$}; 
    \fill[blue] (0,2.5) circle (2pt) node[left] {$5\pi i$}; 
        \fill[blue] (0,0.5) circle (2pt) node[left] {$\pi i$}; 
    \draw[thick, blue, mid arrow](0,0.5)--(1.9,0.5);
    \draw[thick, blue, mid arrow](1.9,2.5)--(0,2.5);
    \draw[thick, blue, mid arrow](1.9,0.5)--(1.9,2.5);
    \node[right] at (-2.7,1.8) {
    ($\log\frac{Q_+}{Q_-}$)-plane};
  \end{scope}

  \begin{scope}[shift={(-2.5,-1.5)}]
    \fill[red] (1,0) circle (2pt) node[above right] {$1$}; 
    \fill[blue] (0,0) circle (2pt) node[left] {$\frac12$}; 
     \draw[thick, blue, mid arrow](0.9,-0.05)--(0.1,-0.05);
  \draw[thick, blue, mid arrow](0.1,0.05)--(0.9,0.05);
\draw[thick,blue] (0.9,-0.05) .. controls (1.25,-0.25) and (1.25,0.25) .. (0.9,0.05);
 \node[left] at (-0.7,0.8) {$z$-plane};
  \end{scope}

  \begin{scope}[shift={(2.5,-1.5)}]
    \draw[->] (-2,0) -- (2,0) node[right] {$\mathrm{Re}$}; 
    \draw[->] (0,-2) -- (0,1) node[above] {$\mathrm{Im}$}; 
    \fill[blue] (0,-1.5) circle (2pt) node[right] {$-3\pi i$}; 
        \fill[blue] (0,0.5) circle (2pt) node[right] {$\pi i$}; 
    \draw[thick, blue, mid arrow](0,0.5)--(-1.9,0.5);
    \draw[thick, blue, mid arrow](-1.9,-1.5)--(0,-1.5);
    \draw[thick, blue, mid arrow](-1.9,0.5)--(-1.9,-1.5);
    \node[left] at (2.7,0.8) {
    ($\log\frac{Q_+}{Q_-}$)-plane};
  \end{scope}
\end{tikzpicture}
\caption{In the above left scope, $z$ starts at $\frac12$, goes along the real axis to a point near $0$, circles around $0$ negatively and goes back to $\frac12$. The corresponding $\log(\frac{Q_+}{Q_-})$ runs along the blue line in the above right scope. As $z$ goes around $1$ in the below left scope, the corresponding trace of $\log(\frac{Q_+}{Q_-})$ is shown in the below right scope. }
\label{figuresanalyticQpm}
\end{figure}
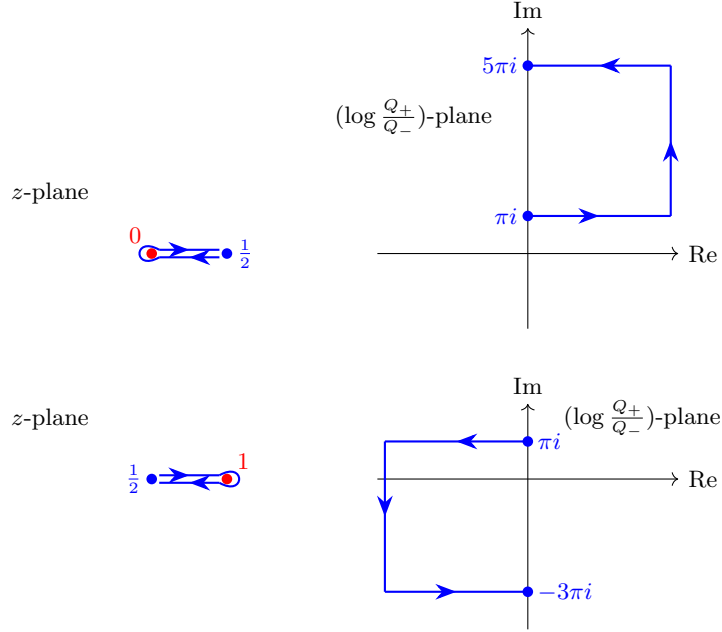

Thus, for the singular set of $\hpsi$, the singular points $\zeta_m^+ = -\log\frac{Q_+}{Q_-} + 2\pi \mathrm{i} m$ shift downward by $4\pi \mathrm{i}$ from left (resp. upward by $4\pi \mathrm{i}$ from right) as $z$ traverses a small circle around $0$ (resp. around $1$). And $\zeta_m^-=2\pi im$ are stable. See the middle column of Figure \ref{figurehpsihphimonodromy}. On the other hand, the singular points of $\hphi$ are $-\zeta_{-m}^{+}=\log\frac{Q_+}{Q_-}+2\pi im$ and $\zeta_m^-=2\pi im$. As $z$ near $\frac12$, by the convention $\log\frac{Q_+}{Q_-}\sim\pi i$, we have $-\zeta_{-m}^+\sim \pi i(2m+1)$. See the right column of Figure \ref{figurehpsihphimonodromy} for the trace of such singular points.
\ \\

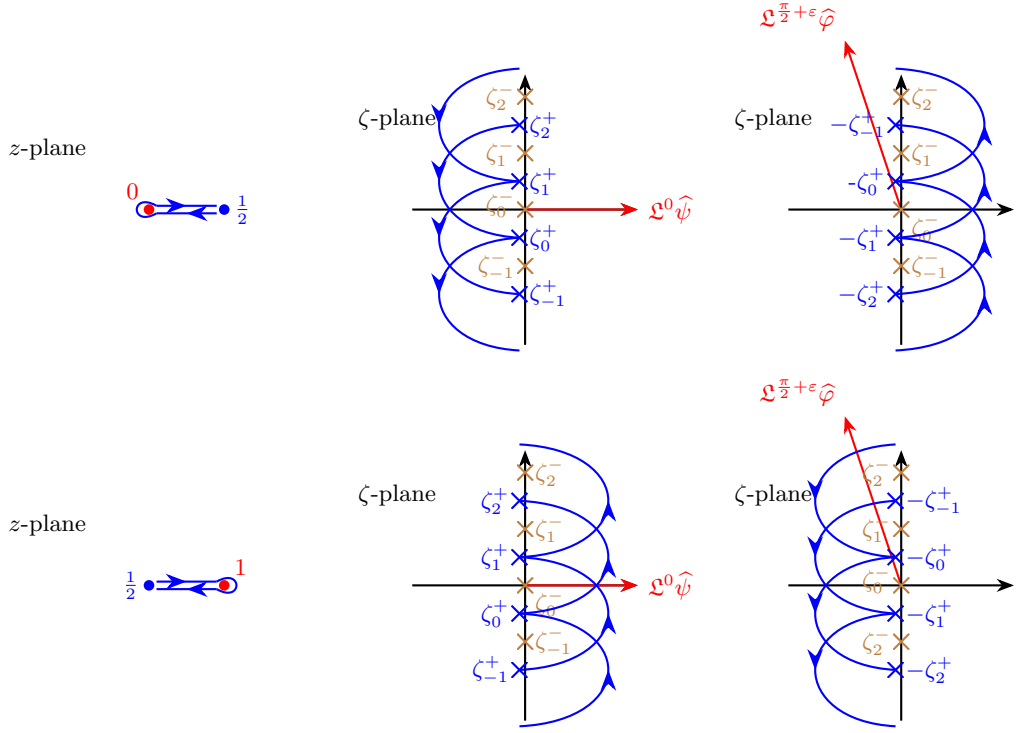
\begin{figure}[h]
    \centering
\begin{tikzpicture}[scale=1, every node/.style={font=\footnotesize},mid arrow/.style={
        postaction={decorate}, 
        decoration={
            markings,
            mark=at position 0.5 with {\arrow{Stealth[scale=1.2]}} 
        }
    }]

  \begin{scope}[shift={(-5,2.5)}]
    \fill[blue] (1,0) circle (2pt) node[right] {$\frac{1}{2}$}; 
    \fill[red] (0,0) circle (2pt) node[above left] {$0$}; 
     \draw[thick, blue, mid arrow](0.9,-0.05)--(0.1,-0.05);
  \draw[thick, blue, mid arrow](0.1,0.05)--(0.9,0.05);
\draw[thick,blue] (0.1,-0.05) .. controls (-0.25,-0.25) and (-0.25,0.25) .. (0.1,0.05);
 \node[left] at (-0.7,0.8) {$z$-plane};
  \end{scope}

  \begin{scope}[shift={(0,2.5)},scale=1.5]
\draw[-{Stealth[scale=1]}, thick] (-1,0)--(1,0); 
  \draw[-{Stealth[scale=1]}, thick,red] (0,0)--(1,0);
  \draw[-{Stealth[scale=1]}, thick] (0,-1.2)--(0,1.2); 
  \node[inner sep=3pt, path picture={
    \draw[thick,blue] (path picture bounding box.south west) -- (path picture bounding box.north east);
    \draw[thick,blue] (path picture bounding box.north west) -- (path picture bounding box.south east);
  }] at (-0.05,-0.25) {};
    \node[inner sep=3pt, path picture={
    \draw[thick,blue] (path picture bounding box.south west) -- (path picture bounding box.north east);
    \draw[thick,blue] (path picture bounding box.north west) -- (path picture bounding box.south east);
  }] at (-0.05,-0.75) {};
    \node[inner sep=3pt, path picture={
    \draw[thick,blue] (path picture bounding box.south west) -- (path picture bounding box.north east);
    \draw[thick,blue] (path picture bounding box.north west) -- (path picture bounding box.south east);
  }] at (-0.05,0.75) {};
  \node[inner sep=3pt, path picture={
    \draw[thick,blue] (path picture bounding box.south west) -- (path picture bounding box.north east);
    \draw[thick,blue] (path picture bounding box.north west) -- (path picture bounding box.south east);
  }] at (-0.05,0.25) {};
    \node[inner sep=3pt, path picture={
    \draw[thick,brown] (path picture bounding box.south west) -- (path picture bounding box.north east);
    \draw[thick,brown] (path picture bounding box.north west) -- (path picture bounding box.south east);
  }] at (0,0) {};
      \node[inner sep=3pt, path picture={
    \draw[thick,brown] (path picture bounding box.south west) -- (path picture bounding box.north east);
    \draw[thick,brown] (path picture bounding box.north west) -- (path picture bounding box.south east);
  }] at (0,0.5) {};
      \node[inner sep=3pt, path picture={
    \draw[thick,brown] (path picture bounding box.south west) -- (path picture bounding box.north east);
    \draw[thick,brown] (path picture bounding box.north west) -- (path picture bounding box.south east);
  }] at (0,1) {};
      \node[inner sep=3pt, path picture={
    \draw[thick,brown] (path picture bounding box.south west) -- (path picture bounding box.north east);
    \draw[thick,brown] (path picture bounding box.north west) -- (path picture bounding box.south east);
  }] at (0,-0.5) {};
  \node[left] at (-0.7,0.8) {$\zeta$-plane};
  \node[right,blue] at (-0.05,0.75) {$\zeta_2^+$};
  \node[right,blue] at (-0.05,0.25) {$\zeta_1^+$};
  \node[right,blue] at (-0.05,-0.25) {$\zeta_0^+$};
  \node[right,blue] at (-0.05,-0.75) {$\zeta_{-1}^+$};
   \node[left,brown] at (0,-0.5) {$\zeta_{-1}^-$};
   \node[left,brown] at (0,0.5) {$\zeta_{1}^-$};
   \node[left,brown] at (0,1) {$\zeta_{2}^-$};
   \node[left,brown] at (0,0.05) {$\zeta_{0}^-$};
   \draw[thick, blue, mid arrow] 
    (-0.05,0.75) .. controls (-1,0.7) and (-1,-0.2) .. (-0.05,-0.25);
       \draw[thick, blue, mid arrow] 
    (-0.05,0.75+0.5) .. controls (-1,0.7+0.5) and (-1,-0.2+0.5) .. (-0.05,-0.25+0.5);
     \draw[thick, blue, mid arrow] 
    (-0.05,0.75-0.5) .. controls (-1,0.7-0.5) and (-1,-0.2-0.5) .. (-0.05,-0.25-0.5);
 \draw[thick, blue, mid arrow] 
    (-0.05,0.75-1) .. controls (-1,0.7-1) and (-1,-0.2-1) .. (-0.05,-0.25-1);
    \node[right,red] at (1,0) {$\mathfrak{L}^0\hpsi$};
  \end{scope}

   \begin{scope}[shift={(5,2.5)},scale=1.5]
\draw[-{Stealth[scale=1]}, thick] (-1,0)--(1,0); 
  \draw[-{Stealth[scale=1]}, thick,red] (0,0)--(-0.5,1.5);
  \draw[-{Stealth[scale=1]}, thick] (0,-1.2)--(0,1.2); 
  \node[inner sep=3pt, path picture={
    \draw[thick,blue] (path picture bounding box.south west) -- (path picture bounding box.north east);
    \draw[thick,blue] (path picture bounding box.north west) -- (path picture bounding box.south east);
  }] at (-0.05,-0.25) {};
    \node[inner sep=3pt, path picture={
    \draw[thick,blue] (path picture bounding box.south west) -- (path picture bounding box.north east);
    \draw[thick,blue] (path picture bounding box.north west) -- (path picture bounding box.south east);
  }] at (-0.05,-0.75) {};
    \node[inner sep=3pt, path picture={
    \draw[thick,blue] (path picture bounding box.south west) -- (path picture bounding box.north east);
    \draw[thick,blue] (path picture bounding box.north west) -- (path picture bounding box.south east);
  }] at (-0.05,0.75) {};
  \node[inner sep=3pt, path picture={
    \draw[thick,blue] (path picture bounding box.south west) -- (path picture bounding box.north east);
    \draw[thick,blue] (path picture bounding box.north west) -- (path picture bounding box.south east);
  }] at (-0.05,0.25) {};
    \node[inner sep=3pt, path picture={
    \draw[thick,brown] (path picture bounding box.south west) -- (path picture bounding box.north east);
    \draw[thick,brown] (path picture bounding box.north west) -- (path picture bounding box.south east);
  }] at (0,0) {};
      \node[inner sep=3pt, path picture={
    \draw[thick,brown] (path picture bounding box.south west) -- (path picture bounding box.north east);
    \draw[thick,brown] (path picture bounding box.north west) -- (path picture bounding box.south east);
  }] at (0,0.5) {};
      \node[inner sep=3pt, path picture={
    \draw[thick,brown] (path picture bounding box.south west) -- (path picture bounding box.north east);
    \draw[thick,brown] (path picture bounding box.north west) -- (path picture bounding box.south east);
  }] at (0,1) {};
      \node[inner sep=3pt, path picture={
    \draw[thick,brown] (path picture bounding box.south west) -- (path picture bounding box.north east);
    \draw[thick,brown] (path picture bounding box.north west) -- (path picture bounding box.south east);
  }] at (0,-0.5) {};
  \node[left] at (-0.7,0.8) {$\zeta$-plane};
  \node[left,blue] at (-0.05,0.75) {$-\zeta_{-1}^+$};
  \node[left,blue] at (-0.05,0.25) {-$\zeta_0^+$};
  \node[left,blue] at (-0.05,-0.25) {$-\zeta_1^+$};
  \node[left,blue] at (-0.05,-0.75) {$-\zeta_{2}^+$};
   \node[right,brown] at (0,-0.5) {$\zeta_{-1}^-$};
   \node[right,brown] at (0,0.5) {$\zeta_{1}^-$};
   \node[right,brown] at (0,1) {$\zeta_{2}^-$};
   \node[below right,brown] at (0,0.05) {$\zeta_{0}^-$};
   \draw[thick, blue, mid arrow] 
    (-0.05,-0.25) .. controls (1,-0.2) and (1,0.7) .. (-0.05,0.75);
       \draw[thick, blue, mid arrow] 
   (-0.05,-0.25+0.5) .. controls (1,-0.2+0.5) and (1,0.7+0.5) .. (-0.05,0.75+0.5) ;
     \draw[thick, blue, mid arrow] 
    (-0.05,-0.25-0.5) .. controls (1,-0.2-0.5) and (1,0.7-0.5) .. (-0.05,0.75-0.5);
 \draw[thick, blue, mid arrow] 
    (-0.05,-0.25-1).. controls (1,-0.2-1)  and (1,0.7-1).. (-0.05,0.75-1) ;
     \node[above left,red] at (-0.5,1.5) {$\mathfrak{L}^{\frac\pi2 + \varepsilon}\hphi$};
  \end{scope}

  \begin{scope}[shift={(-5,-2.5)}]
    \fill[red] (1,0) circle (2pt) node[above right] {$1$}; 
    \fill[blue] (0,0) circle (2pt) node[left] {$\frac12$}; 
     \draw[thick, blue, mid arrow](0.9,-0.05)--(0.1,-0.05);
  \draw[thick, blue, mid arrow](0.1,0.05)--(0.9,0.05);
\draw[thick,blue] (0.9,-0.05) .. controls (1.25,-0.25) and (1.25,0.25) .. (0.9,0.05);
 \node[left] at (-0.7,0.8) {$z$-plane};
  \end{scope}

  \begin{scope}[shift={(0,-2.5)},scale=1.5]
\draw[-{Stealth[scale=1]}, thick] (-1,0)--(1,0); 
  \draw[-{Stealth[scale=1]}, thick,red] (0,0)--(1,0);
  \draw[-{Stealth[scale=1]}, thick] (0,-1.2)--(0,1.2); 
  \node[inner sep=3pt, path picture={
    \draw[thick,blue] (path picture bounding box.south west) -- (path picture bounding box.north east);
    \draw[thick,blue] (path picture bounding box.north west) -- (path picture bounding box.south east);
  }] at (-0.05,-0.25) {};
    \node[inner sep=3pt, path picture={
    \draw[thick,blue] (path picture bounding box.south west) -- (path picture bounding box.north east);
    \draw[thick,blue] (path picture bounding box.north west) -- (path picture bounding box.south east);
  }] at (-0.05,-0.75) {};
    \node[inner sep=3pt, path picture={
    \draw[thick,blue] (path picture bounding box.south west) -- (path picture bounding box.north east);
    \draw[thick,blue] (path picture bounding box.north west) -- (path picture bounding box.south east);
  }] at (-0.05,0.75) {};
  \node[inner sep=3pt, path picture={
    \draw[thick,blue] (path picture bounding box.south west) -- (path picture bounding box.north east);
    \draw[thick,blue] (path picture bounding box.north west) -- (path picture bounding box.south east);
  }] at (-0.05,0.25) {};
    \node[inner sep=3pt, path picture={
    \draw[thick,brown] (path picture bounding box.south west) -- (path picture bounding box.north east);
    \draw[thick,brown] (path picture bounding box.north west) -- (path picture bounding box.south east);
  }] at (0,0) {};
      \node[inner sep=3pt, path picture={
    \draw[thick,brown] (path picture bounding box.south west) -- (path picture bounding box.north east);
    \draw[thick,brown] (path picture bounding box.north west) -- (path picture bounding box.south east);
  }] at (0,0.5) {};
      \node[inner sep=3pt, path picture={
    \draw[thick,brown] (path picture bounding box.south west) -- (path picture bounding box.north east);
    \draw[thick,brown] (path picture bounding box.north west) -- (path picture bounding box.south east);
  }] at (0,1) {};
      \node[inner sep=3pt, path picture={
    \draw[thick,brown] (path picture bounding box.south west) -- (path picture bounding box.north east);
    \draw[thick,brown] (path picture bounding box.north west) -- (path picture bounding box.south east);
  }] at (0,-0.5) {};
  \node[left] at (-0.7,0.8) {$\zeta$-plane};
  \node[left,blue] at (-0.05,0.75) {$\zeta_2^+$};
  \node[left,blue] at (-0.05,0.25) {$\zeta_1^+$};
  \node[left,blue] at (-0.05,-0.25) {$\zeta_0^+$};
  \node[left,blue] at (-0.05,-0.75) {$\zeta_{-1}^+$};
   \node[right,brown] at (0,-0.5) {$\zeta_{-1}^-$};
   \node[right,brown] at (0,0.5) {$\zeta_{1}^-$};
   \node[right,brown] at (0,1) {$\zeta_{2}^-$};
   \node[below right,brown] at (0,0.05) {$\zeta_{0}^-$};
   \draw[thick, blue, mid arrow] 
    (-0.05,-0.25) .. controls (1,-0.2) and (1,0.7) .. (-0.05,0.75);
       \draw[thick, blue, mid arrow] 
   (-0.05,-0.25+0.5) .. controls (1,-0.2+0.5) and (1,0.7+0.5) .. (-0.05,0.75+0.5) ;
     \draw[thick, blue, mid arrow] 
    (-0.05,-0.25-0.5) .. controls (1,-0.2-0.5) and (1,0.7-0.5) .. (-0.05,0.75-0.5);
 \draw[thick, blue, mid arrow] 
    (-0.05,-0.25-1).. controls (1,-0.2-1)  and (1,0.7-1).. (-0.05,0.75-1) ;
    \node[right,red] at (1,0) {$\mathfrak{L}^0\hpsi$};
  \end{scope}

  \begin{scope}[shift={(5,-2.5)},scale=1.5]
\draw[-{Stealth[scale=1]}, thick] (-1,0)--(1,0); 
  \draw[-{Stealth[scale=1]}, thick,red] (0,0)--(-0.5,1.5);
  \draw[-{Stealth[scale=1]}, thick] (0,-1.2)--(0,1.2); 
  \node[inner sep=3pt, path picture={
    \draw[thick,blue] (path picture bounding box.south west) -- (path picture bounding box.north east);
    \draw[thick,blue] (path picture bounding box.north west) -- (path picture bounding box.south east);
  }] at (-0.05,-0.25) {};
    \node[inner sep=3pt, path picture={
    \draw[thick,blue] (path picture bounding box.south west) -- (path picture bounding box.north east);
    \draw[thick,blue] (path picture bounding box.north west) -- (path picture bounding box.south east);
  }] at (-0.05,-0.75) {};
    \node[inner sep=3pt, path picture={
    \draw[thick,blue] (path picture bounding box.south west) -- (path picture bounding box.north east);
    \draw[thick,blue] (path picture bounding box.north west) -- (path picture bounding box.south east);
  }] at (-0.05,0.75) {};
  \node[inner sep=3pt, path picture={
    \draw[thick,blue] (path picture bounding box.south west) -- (path picture bounding box.north east);
    \draw[thick,blue] (path picture bounding box.north west) -- (path picture bounding box.south east);
  }] at (-0.05,0.25) {};
    \node[inner sep=3pt, path picture={
    \draw[thick,brown] (path picture bounding box.south west) -- (path picture bounding box.north east);
    \draw[thick,brown] (path picture bounding box.north west) -- (path picture bounding box.south east);
  }] at (0,0) {};
      \node[inner sep=3pt, path picture={
    \draw[thick,brown] (path picture bounding box.south west) -- (path picture bounding box.north east);
    \draw[thick,brown] (path picture bounding box.north west) -- (path picture bounding box.south east);
  }] at (0,0.5) {};
      \node[inner sep=3pt, path picture={
    \draw[thick,brown] (path picture bounding box.south west) -- (path picture bounding box.north east);
    \draw[thick,brown] (path picture bounding box.north west) -- (path picture bounding box.south east);
  }] at (0,1) {};
      \node[inner sep=3pt, path picture={
    \draw[thick,brown] (path picture bounding box.south west) -- (path picture bounding box.north east);
    \draw[thick,brown] (path picture bounding box.north west) -- (path picture bounding box.south east);
  }] at (0,-0.5) {};
  \node[left] at (-0.7,0.8) {$\zeta$-plane};
  \node[right,blue] at (-0.05,0.75) {$-\zeta_{-1}^+$};
  \node[right,blue] at (-0.05,0.25) {$-\zeta_0^+$};
  \node[right,blue] at (-0.05,-0.25) {$-\zeta_1^+$};
  \node[right,blue] at (-0.05,-0.75) {$-\zeta_{2}^+$};
   \node[left,brown] at (0,-0.5) {$\zeta_{2}^-$};
   \node[left,brown] at (0,0.5) {$\zeta_{1}^-$};
   \node[left,brown] at (0,1) {$\zeta_{2}^-$};
   \node[left,brown] at (0,0.05) {$\zeta_{0}^-$};
   \draw[thick, blue, mid arrow] 
    (-0.05,0.75) .. controls (-1,0.7) and (-1,-0.2) .. (-0.05,-0.25);
       \draw[thick, blue, mid arrow] 
    (-0.05,0.75+0.5) .. controls (-1,0.7+0.5) and (-1,-0.2+0.5) .. (-0.05,-0.25+0.5);
     \draw[thick, blue, mid arrow] 
    (-0.05,0.75-0.5) .. controls (-1,0.7-0.5) and (-1,-0.2-0.5) .. (-0.05,-0.25-0.5);
 \draw[thick, blue, mid arrow] 
    (-0.05,0.75-1) .. controls (-1,0.7-1) and (-1,-0.2-1) .. (-0.05,-0.25-1);
   \node[above left,red] at (-0.5,1.5) {$\mathfrak{L}^{\frac\pi2 + \varepsilon}\hphi$};
  \end{scope}
\end{tikzpicture}
\caption{In the first row, as $z$ goes around $0$ negatively in top left picture, the singular points of $\hpsi$ move along the blue curves from $\zeta_m^+$ to $\zeta_{m-2}^+$ in the top middle picture. Correspondingly, the singular points of $\hphi$ move along the blue curves from $-\zeta_{-m}^+$ to $-\zeta_{-(m+2)}^+$, as shown in the top right picture. In the second row, as $z$ goes around $1$ negatively in the bottom left picture, the singular points of $\hpsi$ move along the blue curves from $\zeta_m^+$ to $\zeta_{m+2}^+$ in the bottom middle picture. Correspondingly, the singular points of $\hphi$ move along the blue curves from $-\zeta_{-m}^+$ to $-\zeta_{-(m-2)}^+$. The function $I_2$ is related to the Laplace transform of $\hphi$ along $0$ direction. The function $I_1$ is related to the Laplace transform of $\hphi$ along direction $\frac\pi2+\varepsilon$.}
\label{figurehpsihphimonodromy}
\end{figure}

Let 
\begin{equation}
    \Delta_{z=p}f(z) := f(z) - {\rm cont}_{z=p}f(z) 
\end{equation}
be the difference of two sheets of the function $f$. We have the following.

\begin{lemma}\label{lemmaMonodromy0} Viewing the Laplace transform $\mathfrak{L}^{\theta} \hphi$ and $\mathfrak{L}^{\theta}\hpsi$ as functions in $z$, we have,
    \begin{equation}\label{equationMonodromy0}
        \Delta_{z=0} \mathfrak{L}^0 \hpsi=0, \quad  \Delta_{z=0} \mathfrak{L}^{\frac\pi2+\varepsilon} \hphi =0, \quad  \Delta_{z=1} \mathfrak{L}^{\frac\pi2+\varepsilon} \hpsi=0, \quad  \Delta_{z=1} \mathfrak{L}^0 \hphi=0.
    \end{equation}
    Moreover, 
    \begin{equation}\label{equationmonodromy0I}
        \Delta_{z=0} I_2 = (1-e^{-4\pi iC})I_2. \quad \Delta_{z=0} I_1=0, \quad \Delta_{z=1} I_1 = \Delta_{z=1} I_2.
    \end{equation}
\end{lemma}

\begin{proof}
    The formulas in equation \eqref{equationMonodromy0} can be read off the trace of the singular points of $\hphi$ and $\hpsi$ in Figure $\ref{figurehpsihphimonodromy}$. For example,  $\Delta_{z=0} \mathfrak{L}^0 \hpsi=0$ is implied by the upper-middle picture of Figure \ref{figurehpsihphimonodromy} since the moving singular points $\zeta_m^+$ run from the left and the Laplace integral is along the positive axis. It is similar to prove the others. 

    Moreover, the first two formulas in equation \eqref{equationmonodromy0I} are implied by \eqref{equationMonodromy0} and the computation in equation \eqref{equationmonodromyQ}. The last formula is by equation \eqref{equationMonodromy0}, \eqref{equationmonodromyQ} and the relation between $I_1$ and $I_2$ (see equation \eqref{equationI2Stokes})
    \begin{equation}
            I_2-I_1 = Q_-^C \frac{e^{2\pi iC}+1}{e^{2\pi iC}-1}\mathfrak{L}^{\frac\pi2 +\varepsilon} \hpsi.
    \end{equation}
\end{proof}

If the trace of singular points crosses the integral line, one should use the curve which avoids the singular points by the Cauchy integral formula. 

\begin{lemma}\label{lemmaDeltaalienhphi} There is a relation between the monodromy operator $\Delta_{z=1}$ and the alien operators:
    \begin{equation}
        \Delta_{z=1} \mathfrak{L}^{\frac\pi2+\varepsilon} \hphi = \mathfrak{L}^{\frac\pi2+\varepsilon }\left( -e^{\zeta_1^+ C}\Delta_{-\zeta_1^+}^+ +e^{\zeta^+_2C}\Delta_{-\zeta_2^+}^+ \right) \hphi
        =
        \left(\frac{Q_-}{Q_+}\right)^C(e^{4\pi iC}-e^{2\pi iC})\mathfrak{L}^{\frac\pi2+\varepsilon } \hpsi
    \end{equation}
\end{lemma}

\begin{proof}
    The second step can be proved by the alien operator computations in equation \eqref{equationalienhphihpsi} and the convention about the commutative between the exponential term and Laplace. For the first step, we need to look into ${\rm cont}_{z=1}\mathfrak{L}^{\frac\pi2+\varepsilon} \hphi$. As $z$ goes around $1$ negatively, the moving singular points push the integral curve \footnote{Indeed, a non-homogeneous vector field can be constructed such that its time-dependent flow coincides with the homotopy used in the Cauchy integral. For technical details, see \cite{AIF_2023__73_5_1987_0}.} as in the left picture in Figure \ref{figurehphimonodromy1}.
    \begin{figure}[h]
    \centering
\begin{tikzpicture}[scale=1, every node/.style={font=\footnotesize},mid arrow/.style={
        postaction={decorate}, 
        decoration={
            markings,
            mark=at position 0.5 with {\arrow{Stealth[scale=1.2]}} 
        }
    }]

  \begin{scope}[shift={(-5,0)},scale=2]
\draw[-{Stealth[scale=1]}, thick] (-1,0)--(1,0); 
  \draw[-{Stealth[scale=1]}, thick] (0,-1.2)--(0,1.2); 
  \node[inner sep=3pt, path picture={
    \draw[thick,blue] (path picture bounding box.south west) -- (path picture bounding box.north east);
    \draw[thick,blue] (path picture bounding box.north west) -- (path picture bounding box.south east);
  }] at (-0.05,-0.25) {};
    \node[inner sep=3pt, path picture={
    \draw[thick,blue] (path picture bounding box.south west) -- (path picture bounding box.north east);
    \draw[thick,blue] (path picture bounding box.north west) -- (path picture bounding box.south east);
  }] at (-0.05,-0.75) {};
    \node[inner sep=3pt, path picture={
    \draw[thick,blue] (path picture bounding box.south west) -- (path picture bounding box.north east);
    \draw[thick,blue] (path picture bounding box.north west) -- (path picture bounding box.south east);
  }] at (-0.05,0.75) {};
  \node[inner sep=3pt, path picture={
    \draw[thick,blue] (path picture bounding box.south west) -- (path picture bounding box.north east);
    \draw[thick,blue] (path picture bounding box.north west) -- (path picture bounding box.south east);
  }] at (-0.05,0.25) {};
    \node[inner sep=3pt, path picture={
    \draw[thick,brown] (path picture bounding box.south west) -- (path picture bounding box.north east);
    \draw[thick,brown] (path picture bounding box.north west) -- (path picture bounding box.south east);
  }] at (0,0) {};
      \node[inner sep=3pt, path picture={
    \draw[thick,brown] (path picture bounding box.south west) -- (path picture bounding box.north east);
    \draw[thick,brown] (path picture bounding box.north west) -- (path picture bounding box.south east);
  }] at (0,0.5) {};
      \node[inner sep=3pt, path picture={
    \draw[thick,brown] (path picture bounding box.south west) -- (path picture bounding box.north east);
    \draw[thick,brown] (path picture bounding box.north west) -- (path picture bounding box.south east);
  }] at (0,1) {};
      \node[inner sep=3pt, path picture={
    \draw[thick,brown] (path picture bounding box.south west) -- (path picture bounding box.north east);
    \draw[thick,brown] (path picture bounding box.north west) -- (path picture bounding box.south east);
  }] at (0,-0.5) {};
  \node[above] at (0,1.2) {$\zeta$-plane};
  \node[right,blue] at (-0.05,0.75) {$-\zeta_{-1}^+$};
  \node[right,blue] at (-0.05,0.25) {$-\zeta_0^+$};
  \node[right,blue] at (-0.05,-0.25) {$-\zeta_1^+$};
  \node[right,blue] at (-0.05,-0.75) {$-\zeta_{2}^+$};
   \node[left,brown] at (0,-0.5) {$\zeta_{-1}^-$};
   \node[left,brown] at (0,0.5) {$\zeta_{1}^-$};
   \node[left,brown] at (0,1) {$\zeta_{2}^-$};
   \node[left,brown] at (0,0.05) {$\zeta_{0}^-$};
   \draw[thick, blue, mid arrow] 
    (-0.05,0.75) .. controls (-1,0.7) and (-1,-0.2) .. (-0.05,-0.25);
     \draw[thick, blue, mid arrow] 
    (-0.05,0.75-0.5) .. controls (-1,0.7-0.5) and (-1,-0.2-0.5) .. (-0.05,-0.25-0.5);
    \draw[thick, red, mid arrow] 
    (0,-0.375) .. controls (-0.2,-0.5) and (-0.2,-0.6) .. (0,-0.625);
     \draw[thick, red] 
    (-0.07,0.14) .. controls (-0.1,-0.15) and (0.2,-0.25) .. (0,-0.375);
    \draw[thick, red] (0,0)--(-0.07,0.14);
    \draw[thick, red] 
    (0,-0.625) .. controls (0.2,-0.7) and (0.2,-0.8) .. (0,-0.875);
    \draw[thick, red] 
    (0,-0.875) .. controls (-0.2,-0.875) and (-0.8,-0.8) .. (-0.9,-0.55);
    \draw[-{Stealth[scale=1]}, thick,red] (-0.9,-0.55)--(-1.5,1.5);
      \node[above,red] at (-1.5,1.5) {${\rm cont}_{z=1}\mathfrak{L}^{\frac\pi2 + \varepsilon}\hphi$};
  \end{scope}
  
  \begin{scope}[shift={(0,0)},scale=2]
\draw[-{Stealth[scale=1]}, thick] (-1,0)--(1,0); 
  \draw[-{Stealth[scale=1]}, thick,red] (0,0)--(-0.5,1.5);
  \draw[-{Stealth[scale=1]}, thick] (0,-1.2)--(0,1.2); 
  \node[inner sep=3pt, path picture={
    \draw[thick,blue] (path picture bounding box.south west) -- (path picture bounding box.north east);
    \draw[thick,blue] (path picture bounding box.north west) -- (path picture bounding box.south east);
  }] at (-0.05,-0.25) {};
    \node[inner sep=3pt, path picture={
    \draw[thick,blue] (path picture bounding box.south west) -- (path picture bounding box.north east);
    \draw[thick,blue] (path picture bounding box.north west) -- (path picture bounding box.south east);
  }] at (-0.05,-0.75) {};
    \node[inner sep=3pt, path picture={
    \draw[thick,blue] (path picture bounding box.south west) -- (path picture bounding box.north east);
    \draw[thick,blue] (path picture bounding box.north west) -- (path picture bounding box.south east);
  }] at (-0.05,0.75) {};
  \node[inner sep=3pt, path picture={
    \draw[thick,blue] (path picture bounding box.south west) -- (path picture bounding box.north east);
    \draw[thick,blue] (path picture bounding box.north west) -- (path picture bounding box.south east);
  }] at (-0.05,0.25) {};
    \node[inner sep=3pt, path picture={
    \draw[thick,brown] (path picture bounding box.south west) -- (path picture bounding box.north east);
    \draw[thick,brown] (path picture bounding box.north west) -- (path picture bounding box.south east);
  }] at (0,0) {};
      \node[inner sep=3pt, path picture={
    \draw[thick,brown] (path picture bounding box.south west) -- (path picture bounding box.north east);
    \draw[thick,brown] (path picture bounding box.north west) -- (path picture bounding box.south east);
  }] at (0,0.5) {};
      \node[inner sep=3pt, path picture={
    \draw[thick,brown] (path picture bounding box.south west) -- (path picture bounding box.north east);
    \draw[thick,brown] (path picture bounding box.north west) -- (path picture bounding box.south east);
  }] at (0,1) {};
      \node[inner sep=3pt, path picture={
    \draw[thick,brown] (path picture bounding box.south west) -- (path picture bounding box.north east);
    \draw[thick,brown] (path picture bounding box.north west) -- (path picture bounding box.south east);
  }] at (0,-0.5) {};
  \node[above] at (0,1.2) {$\zeta$-plane};
  \node[right,blue] at (-0.05,0.75) {$-\zeta_{-1}^+$};
  \node[right,blue] at (-0.05,0.25) {$-\zeta_0^+$};
  \node[right,blue] at (-0.05,-0.25) {$-\zeta_1^+$};
  \node[right,blue] at (-0.05,-0.75) {$-\zeta_{2}^+$};
   \node[left,brown] at (0,-0.5) {$\zeta_{-1}^-$};
   \node[left,brown] at (0,0.5) {$\zeta_{1}^-$};
   \node[left,brown] at (0,1) {$\zeta_{2}^-$};
   \node[left,brown] at (0,0.05) {$\zeta_{0}^-$};
   \node[above,red] at (-0.5,1.5) {$\mathfrak{L}^{\frac\pi2 + \varepsilon}\hphi$};
    \draw[thick, red, mid arrow] 
    (0,-0.375) .. controls (-0.2,-0.5) and (-0.2,-0.6) .. (0,-0.625);
     \draw[thick, red, mid arrow] 
    (0,-0.375) .. controls (-0.2,-0.5) and (-0.2,-0.6) .. (0,-0.625);
        \draw[thick, red, mid arrow] 
    (0,-0.375) .. controls (-0.2,-0.5) and (-0.2,-0.6) .. (0,-0.625);
     \draw[thick, red] 
    (-0.07,0.14) .. controls (-0.1,-0.15) and (0.2,-0.25) .. (0,-0.375);
    \draw[thick, red] (0,0)--(-0.07,0.14);
    \draw[thick, red] 
    (0,-0.625) .. controls (0.2,-0.7) and (0.2,-0.8) .. (0,-0.875);
    \draw[thick, red] 
    (0,-0.875) .. controls (-0.2,-0.875) and (-0.23,-0.65) .. (-0.3,-0.5);
    \draw[-{Stealth[scale=1]}, thick,red] (-0.3,-0.5)--(-1,1.5);
      \node[above left,red] at (-1,1.5) {${\rm cont}_{z=1}\mathfrak{L}^{\frac\pi2 + \varepsilon}\hphi$};
  \end{scope}

    \begin{scope}[shift={(5,0)},scale=2]
\draw[-{Stealth[scale=1]}, thick] (-1,0)--(1,0); 
  \draw[-{Stealth[scale=1]}, thick] (0,-1.2)--(0,1.2); 
  \node[inner sep=3pt, path picture={
    \draw[thick,blue] (path picture bounding box.south west) -- (path picture bounding box.north east);
    \draw[thick,blue] (path picture bounding box.north west) -- (path picture bounding box.south east);
  }] at (-0.05,-0.25) {};
    \node[inner sep=3pt, path picture={
    \draw[thick,blue] (path picture bounding box.south west) -- (path picture bounding box.north east);
    \draw[thick,blue] (path picture bounding box.north west) -- (path picture bounding box.south east);
  }] at (-0.05,-0.75) {};
    \node[inner sep=3pt, path picture={
    \draw[thick,blue] (path picture bounding box.south west) -- (path picture bounding box.north east);
    \draw[thick,blue] (path picture bounding box.north west) -- (path picture bounding box.south east);
  }] at (-0.05,0.75) {};
  \node[inner sep=3pt, path picture={
    \draw[thick,blue] (path picture bounding box.south west) -- (path picture bounding box.north east);
    \draw[thick,blue] (path picture bounding box.north west) -- (path picture bounding box.south east);
  }] at (-0.05,0.25) {};
    \node[inner sep=3pt, path picture={
    \draw[thick,brown] (path picture bounding box.south west) -- (path picture bounding box.north east);
    \draw[thick,brown] (path picture bounding box.north west) -- (path picture bounding box.south east);
  }] at (0,0) {};
      \node[inner sep=3pt, path picture={
    \draw[thick,brown] (path picture bounding box.south west) -- (path picture bounding box.north east);
    \draw[thick,brown] (path picture bounding box.north west) -- (path picture bounding box.south east);
  }] at (0,0.5) {};
      \node[inner sep=3pt, path picture={
    \draw[thick,brown] (path picture bounding box.south west) -- (path picture bounding box.north east);
    \draw[thick,brown] (path picture bounding box.north west) -- (path picture bounding box.south east);
  }] at (0,1) {};
      \node[inner sep=3pt, path picture={
    \draw[thick,brown] (path picture bounding box.south west) -- (path picture bounding box.north east);
    \draw[thick,brown] (path picture bounding box.north west) -- (path picture bounding box.south east);
  }] at (0,-0.5) {};
  \node[above] at (0,1.2) {$\zeta$-plane};
  \node[right,blue] at (-0.05,0.75) {$-\zeta_{-1}^+$};
  \node[right,blue] at (-0.05,0.25) {$-\zeta_0^+$};
  \node[right,blue] at (-0.05,-0.25) {$-\zeta_1^+$};
  \node[right,blue] at (-0.05,-0.75) {$-\zeta_{2}^+$};
   \node[left,brown] at (0,-0.5) {$\zeta_{-1}^-$};
   \node[left,brown] at (0,0.5) {$\zeta_{1}^-$};
   \node[left,brown] at (0,1) {$\zeta_{2}^-$};
   \node[left,brown] at (0,0.05) {$\zeta_{0}^-$};
  \draw[thick, red, mid arrow](-0.69,1.31-1)--(-0.09,0.71-1);
  \draw[thick, red, mid arrow](0.01,0.75-1)--(-0.59,1.35-1);
\draw[thick,red] (-0.09,0.71-1) .. controls (-0.02,0.63-1) and (0.03,0.73-1) .. (0.01,0.75-1);

  \draw[thick, red, mid arrow](-0.69,1.31-1.5)--(-0.09,0.71-1.5);
  \draw[thick, red, mid arrow](0.01,0.75-1.5)--(-0.59,1.35-1.5);
\draw[thick,red] (-0.09,0.71-1.5) .. controls (-0.02,0.63-1.5) and (0.03,0.73-1.5) .. (0.01,0.75-1.5);
\node[above,red] at (-0.69,1.31-1) {$\gamma_1$};
\node[above,red] at (-0.69,1.31-1.5) {$\gamma_2$};
  \end{scope}
\end{tikzpicture}
\caption{Computation of $\Delta_{z=1} \mathfrak{L}^{\frac\pi2+\varepsilon} \hphi$}
\label{figurehphimonodromy1}
\end{figure}
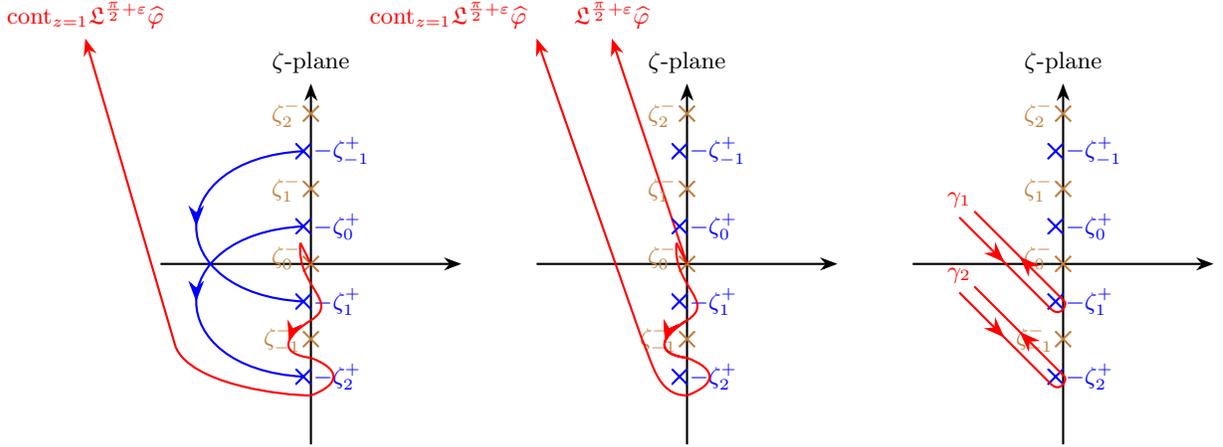

We then compute the difference of two Laplace transforms as shown in the middle picture. The result is 
\begin{equation}
     \Delta_{z=1} \mathfrak{L}^{\frac\pi2+\varepsilon} \hphi =   \mathfrak{L}^{\frac\pi2+\varepsilon}\hphi -{\rm cont}_{z=1}\mathfrak{L}^{\frac\pi2+\varepsilon} \hphi = \left(-\int_{\gamma_1} + 
    \int_{\gamma_2}\right) e^{-C\zeta} \hphi d\zeta,
\end{equation}
where the paths $\gamma_1$ and $\gamma_2$ are shown in the right picture. The symbols we used in the last term of above formula are by the analytic continuation of germ $\hphi$ at $\zeta^-_m$ for all $m\in\mathbb{Z}$: ${\rm cont}_{\zeta_m^-}\hphi= -\hphi$. We get the desired result.
\end{proof}

By using Lemma \ref{lemmamonodromyQ}, Lemma \ref{lemmaMonodromy0}, and Lemma \ref{lemmaalienhphihpsi}, we have the monodromy of $I_1$ and $I_2$.

\begin{proposition}\label{pro:I-monodormy}
    \begin{equation}
        \begin{split}
            \Delta_{z=0} 
\begin{bmatrix}
I_1 \\
I_2 
\end{bmatrix}
&=
\begin{bmatrix}
0&0 \\
0& 1-e^{-4\pi iC} 
\end{bmatrix}
\begin{bmatrix}
I_1 \\
I_2 
\end{bmatrix},
\\
            \Delta_{z=1} 
\begin{bmatrix}
I_1 \\
I_2 
\end{bmatrix} &=  \frac{1}{e^{6\pi iC} +e^{4\pi iC} } 
\begin{bmatrix}
    (e^{6\pi iC}-1) & e^{4\pi iC} -e^{2\pi iC}\\
    (e^{6\pi iC}-1) &e^{4\pi iC} -e^{2\pi iC} 
\end{bmatrix}
\begin{bmatrix}
I_1 \\
I_2 
\end{bmatrix}.
\end{split}
    \end{equation}
\end{proposition}

\section{Conclusion and discussion}
\label{sec:conclusion}
In this paper, we have developed the resurgence analysis for large central charge (large $C$) expansions in two-dimensional CFTs using the Coulomb gas formalism. In particular, we have focused on the conformal blocks $I_1(C,z)$ and $I_2(C,z)$ of the correlation
function involving degenerate primary operator $\phi_{2,1}$, $ \langle \phi_{2,1}(0)\phi_{2,1}(z,\bar{z})\phi_{2,1}(1)\phi_{2,1}(\infty)\rangle$.
Through the exact Borel-Laplace representations of conformal blocks $I_1(C,z)$ and $I_2(C,z)$, we have demonstrated that $I_1(C,z)$ participates in the Stokes phenomenon of $I_2(C,z)$ (and vice versa). Given a certain conformal block, resurgence theory enables us to discover other internal operators (conformal blocks) \cite{Benjamin:2023uib,Bissi:2024wur}.

Since $I_1(C,z)$ and $I_2(C,z)$ can be expressed using the product of Gamma functions and hypergeometric function, one could extract the resurgent data of hypergeometric functions via the formula \eqref{equationalienleibniz}, which relates to the results in \cite{Bissi:2024wur}. Notably, the BPZ equation of $I_1(C,z)$ and $I_2(C,z)$ (manifest as a hypergeometric differential equation) can be rewritten as a Schr\"odinger equation, where $1/C$ plays the role of Planck constant. One can identify $I_1(C,z)$ and $I_2(C,z)$ with the Borel resummed WKB wavefunction \cite{ATT21}, where our results \eqref{equationI2Stokes} and \eqref{equationStokesI1} play a curious role. Moreover, we have shown that monodromy in \( z \) arises from alien calculus in the Borel variable \( \zeta \) (dual to \( C \)), demonstrating the power of resurgence techniques. In contrast, Ecalle's parametric resurgence (or co-equational resurgence) studies this BPZ equation and establishes the relationships between resurgence phenomena as \( C \to \infty \) and \( z \to \infty \), such as those described by the first, second, and third bridge equations \cite{ecalle1994weighted}. This will be an interesting topic in future work.

In our paper, we have seen the resurgence structure of conformal blocks $I_1(C,z)$ and $I_2(C,z)$. A key next step is constructing the four-point function by
combining these conformal blocks with their anti-holomorphic counterparts. It is also interesting to study this norm structure within the framework of resurgence theory, where the crossing symmetry ($z\to 1-z$) and monodromy-free condition (around $z=0,1$) of the correlation function play an essential role.  These properties are closely related to proposition \ref{pro:Izto1-z} and proposition \ref{pro:I-monodormy}.
So far, we have focused on the non-unitary minimal model with a large central charge. Expanding this work to explore the resurgence structure in unitary CFT is particularly important. One interesting direction could be the Liouville theory. Furthermore, the Borel-Laplace representations play a significant role in various domains of  QFT, see \cite{tHooft:1977xjm} for instance. Exploring the application of our method to these representations could yield valuable insights into the study of the non-perturbative method.

From a purely mathematical perspective, elementary solution formulas exist for the roots of the equation $Q(w) = 0$ when $Q$ is a cubic polynomial. However, for polynomials of degree five or higher ($\deg Q \geq 5$), Galois theory establishes that no general solution expressible by radicals can exist. In lieu of explicit algebraic solutions, asymptotic approaches such as the saddle-point method applied to contour integrals or WKB analysis of associated differential equations provide viable alternatives. The resurgence properties of the integral
\begin{equation}
    \int_{\Gamma} Q^C(w)  dw, \quad \text{where } Q \text{ is an arbitrary polynomial in } w
\end{equation}
in the large-$C$ asymptotic regime will be investigated in subsequent research.

\subsection*{Acknowledgements}
We would like to thank Jie Gu, Xia Gu, Si Li, David Sauzin, Shanzhong SUN, Ruidong Zhu and Hao Zou for useful discussions. We would also like to thank Feiyu Peng and Hao Ouyang for their collaboration on the related topic. Y.L. thanks Zhengzhou University for its hospitality during his visit. Y.L. is partially supported by BNSFC No.JR25001. %
H.S. is supported by the National Natural Science Foundation of China (Grant No.12405087), Henan Postdoc Foundation (Grant No.22120055) and the Startup Funding of Zhengzhou University (Grant No.121-35220049, 121-35220581). H.S would like to thank BIMSA, SIMIS and Tongji University for their hospitality.

\appendix
\section{Analytic continuation of the function $\arccos$}\label{Appendixarccos}
\renewcommand{\theequation}{A.\arabic{equation}} 
\setcounter{equation}{0} 

In this appendix, we mainly review the analytic continuation of the function $\arccos$. Let $\widehat\phi$ be a germ holomorphic in a disc $\mathbb{D}$ and let point $p$ be a possibly singular point located at the boundary of $\mathbb{D}$: $p\in\partial\mathbb{D}$. If the germ $\widehat{\phi}$ admits the analytic continuation along a path around $p$ with negative direction, we then denote the final germ by ${\rm cont}_p \widehat\phi$.

\begin{center}
\begin{tikzpicture}[
    mid arrow/.style={
        postaction={decorate}, 
        decoration={
            markings,
            mark=at position 0.5 with {\arrow{Stealth[scale=1.2]}} 
        }
    }
]

    \filldraw[fill=gray!30] (0,0) circle (0.5); 
    
    \foreach \point/\pos in {(0.5,0)/right,(-1.5,0.7)/left,(-1,-0.6)/left,(1.4,-1)/right}
        \draw[red,line width=0.8pt] \point ++(-2pt,-2pt) -- ++(4pt,4pt) ++(-4pt,0) -- ++(4pt,-4pt);
    \coordinate (A) at (-0.25,0); 
  \node[right,red] at (0.6,0) {$p$};
\draw[thick, red, mid arrow] 
    (0,0) .. controls (1.5,2) and (1.5,-2) .. (0.2,-0.2);
\end{tikzpicture}
\end{center}

For instance, let $\mathbb{D}$ be the open unit disc. Then
\begin{equation}
    {\rm cont}_1\frac{1}{x-1}=\frac{1}{x-1},\quad
    {\rm cont}_1\big[(x-1)^{\frac{n}2}\big]=e^{\frac{\pi in}2}(x-1)^{\frac{n}2}, \quad
{\rm cont}_1\big[\log(x-1)\big]=\log(x-1)-2\pi i.
\end{equation}

For the function $\arccos$, we first introduce an elementary formula.

\begin{lemma}
\begin{equation}\label{equationarccos}
    \arccos(x) = \frac1i \log\big(x+(x^2-1)^{\frac12}\big)
\end{equation}
    is holomorphic in the unit disc. Moreover, 
    \begin{equation}\label{equationarccos-x}
        \arccos (x) + \arccos (-x) =\pi.
    \end{equation}
\end{lemma}

\begin{proof}
    The formula \eqref{equationarccos-x} is implied by the formula \eqref{equationarccos} directly. The Let $y=\arccos(x)$. Thus $\cos y=\frac12(e^{iy}+e^{-iy}) = x$. By quadratic formula, it follows that $e^{iy}=x \pm (x^2-1)^{\frac12}$. Thus, $y$ has two possibly expressions $\frac1i\log(x\pm (x^2-1)^{\frac12})$. Recall by definition $\arccos(0)=\frac\pi2$, we get the desired result.
\end{proof}

By equation \eqref{equationarccos}, we have useful analytic information for the function $\arccos$.
\begin{lemma}\label{lemmaarccosanalyticcontinuation}
    The function $\arccos$ defines a holomorphic germ at the origin. It is a resurgent function\footnote{We say $f$ is resurgent function in $\widehat{\mathcal{R}}_{\Omega}$ if  $f$ admits analytic continuation along any path which avoids a discrete closed subset $\Omega$ in $\mathbb{C}$.} in $\widehat{\mathcal{R}}_{\pm1}$ with double branches integrable singular points at $\pm1$:
    \begin{equation}
        \arccos(x) = -i\sqrt2 (x-1)^{\frac12} + O((x-1)^{\frac32}),\quad \arccos(x) = \pi-\sqrt2 (x+1)^{\frac12} + O((x+1)^{\frac32}).
    \end{equation}
    and 
    \begin{equation}
    ({\rm cont}_1 \arccos)(\zeta)= -\arccos(\zeta), \quad
    ({\rm cont}_{-1} \arccos)(\zeta) = -\arccos(\zeta) + 2\pi.
\end{equation}
\end{lemma}

\begin{proof}
    It is clear that $\arccos \zeta$ has singular points exclusively at $\zeta = \pm 1$, where it exhibits a double-branching behavior. Moreover, one can show that the term in the $\log$ function $\zeta+(\zeta^2-1)^{\frac12}$ moves from $i$ to $-i$ with $\mathfrak{Re}(\zeta+(\zeta^2-1)^{\frac12})>0$ (except $\zeta=0$) if $\zeta$ moves along $1+e^{2\pi is}$ with $s$ changing from $\frac12$ to $-\frac12$. This implies that 
    \begin{equation}
        \begin{split}
            \arccos \zeta + {\rm cont}_1 \arccos(x) 
            &= 
            \frac1i \log\big(\zeta+(\zeta^2-1)^{\frac12}\big) +\frac1i \log\big(\zeta-(\zeta^2-1)^{\frac12}\big)
            = 
            0.
        \end{split}
    \end{equation}
On the other side, if $\zeta$ moves along $-1+e^{2\pi is}$ as $s$ changes from $0$ to $-1$, then $\zeta+(\zeta^2-1)^{\frac12}$ goes from $i$ to $-i$ with real part less than $0$ except the two ending points. See the following picture.

\begin{center}
\begin{tikzpicture}[scale=2]

\begin{scope}[xshift=-1cm, local bounding box=left]
  \draw[-{Stealth[scale=1]}, thick] (-0.8,0)--(0.8,0); 
  \draw[-{Stealth[scale=1]}, thick] (0,-0.8)--(0,0.8); 
  \node[inner sep=3pt, path picture={
    \draw[thick] (path picture bounding box.south west) -- (path picture bounding box.north east);
    \draw[thick] (path picture bounding box.north west) -- (path picture bounding box.south east);
  }] at (-0.5,0) {};
  \node[inner sep=3pt, path picture={
    \draw[thick] (path picture bounding box.south west) -- (path picture bounding box.north east);
    \draw[thick] (path picture bounding box.north west) -- (path picture bounding box.south east);
  }] at (0.5,0) {};
  \draw[->,thick,red] (0,0) arc[start angle=0, end angle=-340, radius=0.5];
  \draw[->,thick,blue] (0,0) arc[start angle=180, end angle=-160, radius=0.5];
  \node[left] at (-0.4,0.8) {$\zeta$-plane};
  \node[below] at (0.5,0) {$1$};
  \node[below] at (-0.5,0) {$-1$};
\end{scope}

\begin{scope}[xshift=2cm, local bounding box=right, dot/.style = {circle, fill, inner sep=0.5pt}]
  \draw[-{Stealth[scale=1]}, thick] (-0.8,0)--(0.8,0); 
  \draw[-{Stealth[scale=1]}, thick] (0,-0.8)--(0,0.8); 
\fill[] (0,0.5) circle (0.8pt) node[above right] {$i$};
\fill[] (0,-0.5) circle (0.8pt) node[below right] {$-i$};
   \node[inner sep=3pt, path picture={
    \draw[thick] (path picture bounding box.south west) -- (path picture bounding box.north east);
    \draw[thick] (path picture bounding box.north west) -- (path picture bounding box.south east);
     }] at (0,0) {};
    \node[below right] at (0,0) {$0$};
    \node[right] at (0.4,0.8) {$[\zeta+(\zeta^2-1)^{\frac12}]$-plane};
      \draw[->,thick, red] 
    (0,0.5) .. controls (-0.2,0) and (-0.2,0) .. (-0.05,-0.45);
          \draw[->,thick, blue] 
    (0,0.5) .. controls (0.6,0) and (0.6,0) .. (0.05,-0.47);
\end{scope}
\end{tikzpicture}
\end{center}
This implies a $2\pi i$ difference of the final result in this lemma by the $\log$ function. 
\end{proof}

\bibliographystyle{JHEP.bst}
\bibliography{BIBLIOGRAPHY}

\end{document}